\documentclass[a4paper]{article}
\usepackage[utf8]{inputenc}
\usepackage[T1]{fontenc}
\usepackage{amsmath}
\usepackage{amssymb}
\usepackage{amsthm}
\usepackage{geometry}
\usepackage{authblk}
\theoremstyle{plain}
\newtheorem{thm}{Theorem} 
\newtheorem{lem}{Lemma}
\newtheorem{cor}{Corollary}

\theoremstyle{definition}
\newtheorem{remark}{Remark}
\newtheorem{example}{Example}
\newtheorem*{thm*}{Theorem}

\usepackage{graphicx,color}

\newcommand{\dif}{\mathop{}\!\mathrm{d}}

\newcommand{\ZZ}{\mathbb{Z}}
\newcommand{\RR}{\mathbb{R}}
\newcommand{\CC}{\mathbb{C}}

\newcommand{\diag}{\mathrm{diag}}
\newcommand{\He}{\mathrm{He}}

\DeclareMathOperator{\tr}{tr}

\newcommand{\beq}{\begin{equation}}
	\newcommand{\eeq}{\end{equation}}
\newcommand{\nn}{\nonumber}
\def\={\;=\;}
\def\+{\,+\,}

\title{On determinantal formulas for hermitian random matrices}
\date{}

\author[1]{Di YANG}
\author[1]{Jiayi ZHAO}
\author[2]{Jian ZHOU}

\affil[1]{
	School of Mathematical Sciences, University of Science and Technology of China \par
	Hefei 230026, P.R. China \par
	\texttt{diyang@ustc.edu.cn},\qquad\texttt{zhaojiayi616@mail.ustc.edu.cn}}
\affil[2]{
	Department of Mathematical Sciences, Tsinghua University \par
	Beijing 100084, P.R. China \par
	\texttt{jianzhou@mail.tsinghua.edu.cn}
}

\begin{document}
\maketitle

\begin{abstract}
	In this paper, we give a direct proof of determinantal formulas for connected $k$-point functions for hermitian matrix models. We also give a new proof of KP integrability for them. From the viewpoint of KP hierarchy, we further give a new proof of the explicit formula for the corresponding affine coordinates. Furthermore, duality for some hermitian matrix models is proved.
\end{abstract}

\tableofcontents

\section{Introduction}

Consider a positive measure $\dif\mu$ over $\RR$. Let $\pi_k(\lambda)=\lambda^k+\cdots$, $k\ge0$, be the associated monic orthogonal polynomials with respect to $\dif\mu$. Denote
	\begin{equation}
		h_k:=\int_\RR\pi_k(\lambda)^2\dif\mu(\lambda),\qquad k\ge0. \label{orth}
	\end{equation}
 The orthogonal polynomials satisfy the following three-term recurrence relations:
	\begin{equation}
		\lambda\pi_l(\lambda)=\pi_{l+1}(\lambda)+\beta_l\pi_l(\lambda)+\gamma_l\pi_{l-1}(\lambda),\quad l\ge1, \label{three-term}
	\end{equation}
	where $\beta_l,\gamma_l\in\RR$. In the 19th century, motivated from numerical integration and approximation theory, the Christoffel--Darboux kernel $K_n(\lambda_1,\lambda_2)$ was introduced:
	\begin{equation}
		K_n(\lambda_1,\lambda_2):=\sum_{l=0}^{n-1}\frac{\pi_l(\lambda_1)\pi_l(\lambda_2)}{h_l},\quad n\ge0,
	\end{equation}
which was shown~\cite{Christoffel,Darboux} to satisfy the following identity: for any $n\ge1$,
	\begin{equation}
		K_n(\lambda_1,\lambda_2)=\frac{1}{h_{n-1}}\frac{\pi_n(\lambda_1)\pi_{n-1}(\lambda_2)-\pi_{n-1}(\lambda_1)\pi_n(\lambda_2)}{\lambda_1-\lambda_2}.\label{c-d}
	\end{equation}

Since the mid-twentieth century, to understand energy level repulsion in complex heavy nuclei, Wigner, Dyson, Mehta and others established the random matrix theory, where the probability distributions of eigenvalues are important subjects. Dyson~\cite{Dyson1962} studied the joint probability density of $k$ eigenvalues of the hermitian matrix with respect to $\dif\mu$
\begin{equation}
	\rho_k(n;\lambda_1,\cdots,\lambda_k)=\frac{\int_{\RR^{n-k}}\prod_{1\le i<j\le n}(\lambda_i-\lambda_j)^2\dif\mu(\lambda_{k+1})\cdots\dif\mu(\lambda_n)}{\int_{\RR^n} \prod_{1\le i<j\le n}(\lambda_i-\lambda_j)^2\dif\mu(\lambda_1)\cdots\dif\mu(\lambda_n)},
\end{equation} and further showed~\cite{Dyson1970} that it can be written as a determinant using the Christoffel--Darboux kernel:
	\begin{equation} \label{dpp}
			\rho_k(n;\lambda_1,\cdots,\lambda_k)=\frac{1}{(n-k+1)_k}\det(K_n(\lambda_i,\lambda_j))_{i,j=1}^{k},
		\end{equation}
	where $(n)_k:=n(n+1)\cdots(n+k-1)$ denotes the increasing Pochhammer symbol.

The notion of correlators was known~\cite{thooft1,thooft2} to have important relations with field theory of strong interactions. For the Gaussian unitary ensemble (GUE), where $\dif\mu(\lambda)=e^{-\frac{\lambda^2}{2}}\dif\lambda$, the $k$-point correlators
\begin{equation}
	\langle\tr M^{i_1}\cdots\tr M^{i_k}\rangle(n):=\frac{\int_{\mathcal H_n}\tr M^{i_1}\cdots\tr M^{i_k}e^{-\frac{1}{2}\tr M^2}\dif M}{\int_{\mathcal H_n}e^{-\frac{1}{2}\tr M^2}\dif M}
\end{equation}
were considered in~\cite{BIPZ,BIZ,HZ} (see also~\cite{Mehta,Deift}). Here $\mathcal H_n$ denotes the space of $n\times n$ hermitian matrices. Denote by $\langle\tr M^{i_1}\cdots\tr M^{i_k}\rangle_c(n)$ the connected $k$-point correlators. Explicit formula for $\langle\tr M^{i}\rangle_c(n)$ was given in~\cite{HZ}. Explicit formula for $\langle\tr M^{i_1}\tr M^{i_2}\rangle_c(n) $ was given in~\cite{MS}. For general $k\ge1$, explicit formula for connected $k$-point GUE correlators was given in~\cite{DY} based on 1D Toda integrability.

It is known~\cite{Shaw} that the partition function of GUE correlators is a KP tau-function. In~\cite{Emergent}, for any KP tau-function $\tau$, the $k$-point generating series of logarithmic derivatives of $\tau$ are expressed by
\begin{align}
	\sum_{j\ge1}\left.\frac{\partial\log\tau}{\partial t_j}\right|_{\mathbf{t}=0}\frac{1}{\lambda_1^{j+1}}=&\lim_{\lambda\to\lambda_1}\left(\widehat A(\lambda,\lambda_1)-\frac{1}{\lambda-\lambda_1}\right),\label{kpmain0}\\
	\sum_{j_1,\cdots,j_k\ge1}\left.\frac{\partial^k\log\tau}{\partial t_{j_1}\cdots\partial t_{j_k}}\right|_{\mathbf{t}=0}\frac{1}{\prod_{i=1}^k\lambda_i^{j_i+1}}=&\frac{(-1)^{k-1}}{k}\sum_{\sigma\in S_k}\prod_{i=1}^k\widehat A(\lambda_{\sigma(i)},\lambda_{\sigma(i+1)})-\frac{\delta_{k,2}}{(\lambda_1-\lambda_2)^2},\,k\ge2,\label{kpmain}
\end{align}
where
\begin{align}
	\widehat A(\lambda_1,\lambda_2):=&\sum_{i,j\ge0}\frac{A_{i,j}}{\lambda_1^{j+1}\lambda_2^{i+1}}+\frac{1}{\lambda_1-\lambda_2}.
\end{align}
Here $A_{i,j}$ are the affine coordinates (cf.~\cite{BY,HE,Sato,ZhouWi,Emergent}) of the element in the Sato Grassmannian corresponding to $\tau$. For GUE, we denote the affine coordinates by $A_{i,j}^{\text{GUE}}$, whose explicit expressions are given by~\cite{Zhou}
\begin{equation}\label{kpgue}
	A_{i,j}^{\text{GUE}}=\begin{cases}
		(-1)^{i+\left[\frac{i+1}{2}\right]}\frac{(2m-1)!!}{(2m)!}\binom{m-1}{\left[\frac{i}{2}\right]}\prod_{k=-i}^{2m-1-i}(n+k),&i+j=2m-1,\\
		0,&\text{otherwise}.
	\end{cases}
\end{equation}
In this case,~\eqref{kpmain0}--\eqref{kpgue} were reproved in~\cite{Toda} using the method of~\cite{DY} and the following
\begin{equation}
	\sum_{i,j\ge0}\frac{A_{i,j}^\text{GUE}}{\lambda_1^{j+1}\lambda_2^{i+1}}+\frac{1}{\lambda_1-\lambda_2}=\frac{1}{\lambda_2\Gamma(n)}\frac{\psi_A^\text{G}(\lambda_1,n)\psi_B^\text{G}(\lambda_2,n-1)-\psi_A^\text{G}(\lambda_1,n-1)\psi_B^\text{G}(\lambda_2,n)}{\lambda_1-\lambda_2}\frac{\lambda_2^n}{\lambda_1^n}
\end{equation}
with
\begin{align}
	\psi_A^\text{G}(\lambda,n)=\sum_{j\ge0}(-1)^j\frac{(n-2j+1)_{2j}}{2^jj!\lambda^{2j}}\lambda^n,\qquad\psi_B^\text{G}(\lambda,n)=\Gamma(n+1)\sum_{j\ge0}\frac{(n+1)_{2j}}{2^jj!\lambda^{2j}}\lambda^{-n}.
\end{align}

For the case when $\dif\mu(\lambda)=w(\lambda)\dif\lambda$ where $w:\RR\to\RR$ is a probability density function with finite moments of all orders, it is known from~\cite{GMO,Martinec,AvM,DY,Toda,Yin,YGueI} that the partition function of the associated hermitian matrix model is a 1D Toda tau-function. Then from~\cite{DY,Toda} we know that generating series of connected $k$-point correlators can be represented by using the following kernel
\begin{equation} \label{dn}
	D_n(\lambda_1,\lambda_2):=\frac{\psi_A(\lambda_1,n)\psi_B(\lambda_2,n-1)-\psi_A(\lambda_1,n-1)\psi_B(\lambda_2,n)}{\lambda_1-\lambda_2},
\end{equation}
where $\psi_A(\lambda,n),\psi_B(\lambda,n)$ are a pair of wave functions of the associated Lax operator $L=\mathcal T+\beta_n+\gamma_n\mathcal T^{-1}$ with $\mathcal T:f(n)\mapsto f(n+1)$ the shifting operator.

For the case when the measure $\dif\mu(\lambda)=e^{V(\lambda)}\dif\lambda$ with $V'(\lambda)$ being a rational function, a formula for the associated connected $k$-point correlators was obtained by Ruzza~\cite{Ruzza} using the Riemann-Hilbert approach, where the Cauchy--Hilbert--Stieltjes transform (cf.~\cite{Deift,Bergere,BE,Ruzza}) of the orthogonal polynomials $\pi_n(\lambda)$
	\begin{equation} \label{chs}
		\widehat{\pi}_n(\xi):=\int_\RR\frac{\pi_n(\lambda)}{\lambda-\xi}\dif\mu(\lambda),\qquad\xi\in\CC\backslash\RR, \, n\in\ZZ_{\ge0},
	\end{equation}
	was considered. Actually, our observation is that the asymptotic expansions of $\pi_n(\lambda)$ and of $-\lambda\widehat\pi_n(\lambda)$ as $\lambda\to\infty$ within any sector in $\CC\backslash\RR$ form a pair of wave functions of the associated Lax operator.

Motivated by the above discussions, for any positive Borel measure $\dif\mu$ over $\RR$ supported on infinitely many points with finite moments of all orders~\cite{Deift,Konig}, denote by $\pi_n(\lambda)$ the orthogonal polynomials with respect to $\dif\mu$, by $\widehat\pi_n(\lambda)$ the 
	Cauchy--Hilbert--Stieltjes transform of $\pi_n(\lambda)$ (see~\eqref{chs}), and by
	\begin{equation}
		\psi_A^\star(\lambda,n)=(1+O(\lambda^{-1}))\lambda^n,\quad\psi_B^\star(\lambda,n)=(1+O(\lambda^{-1}))h_n\lambda^{-n}, \qquad n\in\mathbb{Z}_{\ge0},
	\end{equation}
	the asymptotic expansions of $\pi_n(\lambda)$, $-\lambda\widehat\pi_n(\lambda)$ as $\lambda\to\infty$ within any sector in $\CC\backslash\RR$. Introduce the kernel
\begin{equation} \label{dstar}
	D_n^\star(\lambda_1,\lambda_2):=\frac{\psi_A^\star(\lambda_1,n)\psi_B^\star(\lambda_2,n-1)-\psi_A^\star(\lambda_1,n-1)\psi_B^\star(\lambda_2,n)}{\lambda_1-\lambda_2},\quad n\ge1.
\end{equation}
We also denote by $\langle\tr M^{i_1}\cdots\tr M^{i_k}\rangle_c(n)$ the connected $k$-point correlators of the associated hermitian matrix model, and let
\begin{equation}
	C_k(n;\lambda_1,\cdots,\lambda_k):=\sum_{i_1,\cdots,i_k\ge1}\frac{\langle\tr M^{i_1}\cdots\tr M^{i_k}\rangle_c(n)}{\lambda_1^{i_1+1}\cdots\lambda_k^{i_k+1}},\qquad n\in\ZZ_{\ge0}.
\end{equation}	

	\begin{thm} \label{main}
		The following formula holds for $n\in\ZZ_{\ge0}$:
		\begin{align} \label{maineq}
			C_k(n;\lambda_1,\cdots,\lambda_k)=\frac{(-1)^{k-1}}{k}\sum_{\sigma\in S_k}\prod_{i=1}^kB_n(\lambda_{\sigma(i)},\lambda_{\sigma(i+1)})-\frac{\delta_{k,2}}{(\lambda_1-\lambda_2)^2},\qquad k\ge2,
		\end{align}
		where $B_n(\lambda_1,\lambda_2):=D_n^\star(\lambda_1,\lambda_2)/(\lambda_2h_{n-1})$ for $n\in\ZZ_{\ge1}$ and $B_0(\lambda_1,\lambda_2):=1/(\lambda_1-\lambda_2)$.
		Moreover, for $k=1$ and $n\in\ZZ_{\ge0}$,
		\begin{equation} 
			C_1(n;\lambda_1)=\lim_{\lambda\to\lambda_1}\left(B_n(\lambda,\lambda_1)-\frac{1}{\lambda-\lambda_1}\right)-\frac{n}{\lambda_1}.\label{maineq1}
		\end{equation}
	\end{thm}

The above formula~\eqref{maineq} in Theorem~\ref{main} was essentially known (perhaps in different generalities) 
	from \cite{DY,Ruzza,Toda,Yin} (see the discussions above).
We will give a new proof of Theorem~\ref{main} in Section~\ref{T1sec} using~\eqref{c-d},~\eqref{dpp}. Our proof may shed some lights on the relationship between the probability theory of determinantal point processes (cf.~\cite{DPP}) and the theory of integrable hierarchies.

	The following lemma is similar to~\cite[Lemma 4.4]{YZ}.

\begin{lem} \label{exp-an}
	For $n\in\ZZ_{\ge0}$, we have the expansion
	\begin{equation}
		\frac{\lambda_2^n}{\lambda_1^n}B_n(\lambda_1,\lambda_2)=\frac{1}{\lambda_1-\lambda_2}+\sum_{k,l\ge0}\frac{A_{k,l}(n)}{\lambda_1^{l+1}\lambda_2^{k+1}}.\label{affine}
	\end{equation}
\end{lem}

Building on Theorem~\ref{main}, Lemma~\ref{exp-an} and Zhou's theorem~\cite{Emergent} (cf.~also~\cite[Corollary 2.2]{YZ},~\cite[Lemma 3.1]{xy}), we will give a new proof of the following result.
\begin{cor}[\cite{HH,Painleve,Yin,Zhou1}] \label{KPthm}
	The partition function of the hermitian matrix model associated to $\dif\mu$ is a KP tau-function.
\end{cor}

The $A_{k,l}(n)$ appeared in~\eqref{affine} are the affine coordinates of the KP tau-function associated to $\dif\mu$. Let $\{m_k\}_{k\ge0}$ be the moments of $\dif\mu$. We introduce the notation
\begin{equation}
	H_{(i_1,\cdots,i_n;j_1,\cdots,j_n)}:=(m_{i_k+j_l})_{k,l=1}^n.
\end{equation}
By expanding formula~\eqref{dstar}, we will give a new proof of the following

	\begin{thm}[\cite{Zhou1} (cf.~\cite{Painleve})] \label{hankel}
		For $j<n$, the affine coordinates $A_{i,j}(n)$ read
		\begin{equation} \label{hk}
			A_{i,j}(n)=\frac{(-1)^j}{\Delta_{n-1}}\det H_{(0,1,\cdots,\widehat{n-j-1},\cdots,n-1,n+i;0,1,\cdots,n-1)},
		\end{equation}
		where $\Delta_{n-1}:=\det H_{(0,1,\cdots,n-1;0,1,\cdots,n-1)}$.
	\end{thm}

	From the perspective of the KP hierarchy, following~\cite{Zhou2,Zhou3}, for two hermitian matrix models $M$ and $\widetilde M$, we say that $\widetilde M$ is a {\it KP dual} of $M$ if for any Schur polynomial $s_\rho$ in $n$ variables associated to the partition $\rho$,
\begin{equation}
	\widetilde{\langle s_\rho\rangle}(n)=(-1)^{|\rho|}\langle s_{\rho'}\rangle(-n),\qquad |n|\ge\max\{2,|\rho|\},
\end{equation}
where $\rho'$ is the conjugate partition of $\rho$. Obviously, $M$ is also a KP dual of $\widetilde M$. The following theorem gives a class of hermitian matrix models KP dual to each other.

\begin{thm} \label{dual-ex}
	For $\beta\in\RR$ and $\gamma>0$, consider hermitian matrix models $M,\widetilde M$ associated respectively to 
	\begin{equation}
		\dif\mu(\lambda)=\frac{1}{2\pi\gamma}\sqrt{4\gamma-(\lambda-\beta)^2}\mathbf{1}_I(\lambda)\dif\lambda,\qquad \dif\widetilde\mu(\lambda)=\frac{1}{\pi}\frac{1}{\sqrt{4\gamma-(\lambda-\beta)^2}}\mathbf{1}_I(\lambda)\dif\lambda,
	\end{equation}
	where $I:=(\beta-2\sqrt{\gamma},\beta+2\sqrt{\gamma})$. Then $\widetilde M$ is a KP dual of $M$.
\end{thm}

We will give some preliminaries in Section~\ref{prel}, and prove Theorem~\ref{main} in Section~\ref{T1sec}. In Section~\ref{sec-kp}, we will give the proof of Corollary~\ref{KPthm} and Theorem~\ref{hankel}. We will also describe the polynomiality of the connected correlators in terms of coefficients of orthogonal polynomials in Section~\ref{section-rational}. We will prove Theorem~\ref{dual-ex} in Section~\ref{sec-dual}.
\newline
\par\noindent{\bf Acknowledgements.} D. Yang would like to thank Giulio Ruzza and Marco Bertola for very helpful discussions. We also thank Don Zagier for helpful discussions. D. Yang and J. Zhao are partially supported by NSFC No. 12371254 and CAS YSBR-032. J. Zhou is partly supported by NSFC No. 11890662.
	
\section{Preliminaries} \label{prel}

For $n\ge0$, we denote the expected value of a function $f:\mathcal{H}_n(E)\to\RR$ by~\cite{Eynard, Mehta}
	\begin{equation}
		\left\langle f(M)\right\rangle(n):=\frac{1}{Z_n}\int_{\mathcal{H}_n(E)}f(M)\dif\nu_n(M),
	\end{equation} where $\dif\nu_n$ is the unitarily invariant measure on $\mathcal{H}_n(E)$ defined via the spectral decomposition $M=U\diag(\lambda_1,\cdots,\lambda_n)U^*$ with respect to the measure $(\dif\mu)^n$ over $E^n$ and the Haar measure over $U(n)$, and
	\begin{equation}
		Z_n:=\int_{\mathcal H_n(E)}\dif\nu_n(M).
	\end{equation}
	Here $Z_0\equiv1$. For $n\ge0$, denote the disconnected $k$-point correlators by $\langle\tr M^{i_1}\cdots\tr M^{i_k}\rangle(n)$, and denote the connected $k$-point correlators $\langle\tr M^{i_1}\cdots\tr M^{i_k}\rangle_c(n)$ by the M\"obius relation~\cite{Rota,Mehta,Emergent}:
	\begin{align}
		\langle\tr M^{i_1}\cdots\tr M^{i_k}\rangle(n)&=\sum_{\text{partition }\mathcal{P}\text{ of }\mathbf Z_k}\prod_{I\in\mathcal{P}}\left\langle\prod_{j\in I}\tr M^{i_j}\right\rangle_c(n),\\
		\langle\tr M^{i_1}\cdots\tr M^{i_k}\rangle_c(n)&=\sum_{\text{partition }\mathcal{P}\text{ of }\mathbf Z_k}(-1)^{|\mathcal P|-1}(|\mathcal P|-1)!\prod_{I\in\mathcal{P}}\left\langle\prod_{j\in I}\tr M^{i_j}\right\rangle(n),
	\end{align}
	where $\mathbf Z_k:=\{1,\cdots,k\}$.

	We denote for $\lambda_1,\cdots,\lambda_k\in\CC\backslash\RR$ and $n\ge0$,
	\begin{align}
		\mathcal C_k(n;\lambda_1,\cdots,\lambda_k):=&\sum_{\text{partition }\mathcal{P}\text{ of }\mathbf Z_k}(-1)^{|\mathcal P|-1}(|\mathcal P|-1)!\prod_{I\in\mathcal{P}}\left\langle\prod_{j\in I}\tr\frac{1}{\lambda_j-M}\right\rangle(n)-\frac{n\delta_{k,1}}{\lambda_1}.
	\end{align}
	It is analytic on $(\CC\backslash\RR)^k$. For $n\ge0$, as $\lambda_1,\cdots,\lambda_k\to\infty$ within any sector in $(\CC\backslash\RR)^k$, we have
	\begin{equation} \label{asym-func}
		\mathcal C_k(n;\lambda_1,\cdots,\lambda_k)\sim C_k(n;\lambda_1,\cdots,\lambda_k).
	\end{equation}
We note that for the case of $\mathbb{CP}^1$, the analytic $k$-point functions were first discussed in~\cite{DYZcp1}.

	It is known~\cite{Ruzza} that for $n\ge0$, as $\lambda\to\infty$ within any sector in $\CC\backslash\RR$,
		\begin{equation} \label{poincare}
			\psi_B^\star(\lambda,n)=\frac{1}{\lambda^n}\sum_{j=0}^\infty \frac{1}{\lambda^j}\int_\RR y^{j+n}\pi_n(y)\dif\mu(y).
		\end{equation}
	In addition, it is known in~\cite{Ruzza,Eynard} that for any $n\ge1$ and $\lambda\in\CC\backslash\RR$,
		\begin{equation}
			\pi_n(\lambda)\widehat\pi_{n-1}(\lambda)-\pi_{n-1}(\lambda)\widehat\pi_n(\lambda)=-h_{n-1}.\label{reduced-cd}
		\end{equation} 
		Denote for $\lambda_1,\lambda_2\in\CC\backslash\RR$,
		\begin{align}
		\mathcal J_n(\lambda_1,\lambda_2):=&-\frac{1}{h_{n-1}}\frac{\pi_n(\lambda_1)\widehat\pi_{n-1}(\lambda_2)-\pi_{n-1}(\lambda_1)\widehat\pi_n(\lambda_2)}{\lambda_1-\lambda_2},\qquad n\ge1, \label{jn}\\
		\mathcal J_0(\lambda_1,\lambda_2):=&\frac{1}{\lambda_1-\lambda_2},\\
		\mathcal J^\circ_n(\lambda_1,\lambda_2):=&\mathcal J_n(\lambda_1,\lambda_2)-\frac{1}{\lambda_1-\lambda_2},\qquad n\ge0.
	\end{align}
	It is known that~\cite{Ruzza,Eynard}
	\begin{align}
		\mathcal J_n^\circ(\lambda_1,\lambda_2)=-\sum_{l=0}^{n-1}\frac{\pi_l(\lambda_1)\widehat\pi_l(\lambda_2)}{h_l}=-\int_\RR\frac{K_n(\lambda_1,\lambda)}{\lambda-\lambda_2}\dif\mu(\lambda),\qquad n\ge0. \label{trans-cd}
	\end{align}
	Obviously, $\mathcal J_n^\circ(\lambda_1,\lambda_2)$ is analytic on $(\CC\backslash\RR)^2$. Similar to~\cite{DYZcp1}, we prove the following

	\begin{lem} \label{analytic}
		For $k\ge2$ and $n\ge0$, the function
		\begin{equation} \label{hkk}
			\frac{(-1)^{k-1}}{k}\sum_{\sigma\in S_k}\prod_{i=1}^k\mathcal J_n(\lambda_{\sigma(i)},\lambda_{\sigma(i+1)})-\frac{\delta_{k,2}}{(\lambda_1-\lambda_2)^2}
		\end{equation}
		is analytic along $\lambda_i=\lambda_j,i\neq j$. Here $\lambda_1,\cdots,\lambda_k\in\CC\backslash\RR$.
	\end{lem}
	\begin{proof}
	Without loss of generality, we show that~\eqref{hkk} is analytic along $\lambda_{k-1}=\lambda_k$. For $k\ge3$, we have
	\begin{align}
		&-\frac{1}{k}\sum_{\sigma\in S_k}\prod_{i=1}^k\mathcal J_n(\lambda_{\sigma(i)},\lambda_{\sigma(i+1)})\nn\\
		=&-\sum_{\sigma\in S_{k-2}}\left(\prod_{i=1}^{k-3}\mathcal J_n(\lambda_{\sigma(i)},\lambda_{\sigma(i+1)})\right)\frac{\mathcal J_n(\lambda_{\sigma(k-2)},\lambda_{k-1})\mathcal J_n(\lambda_k,\lambda_{\sigma(1)})-\mathcal J_n(\lambda_{\sigma(k-2)},\lambda_k)\mathcal J_n(\lambda_{k-1},\lambda_{\sigma(1)})}{\lambda_{k-1}-\lambda_k}\nn\\
		&-\sum_{\sigma\in S_{k-2}}\left(\prod_{i=1}^{k-3}\mathcal J_n(\lambda_{\sigma(i)},\lambda_{\sigma(i+1)})\right)\bigg(\mathcal J_n(\lambda_{\sigma(k-2)},\lambda_{k-1})\mathcal J_n^\circ(\lambda_{k-1},\lambda_k)\mathcal J_n(\lambda_k,\lambda_{\sigma(1)})\nn\\
		&\qquad\qquad\qquad\qquad\qquad\qquad\qquad\qquad+\mathcal J_n(\lambda_{\sigma(k-2)},\lambda_k)\mathcal J_n^\circ(\lambda_k,\lambda_{k-1})\mathcal J_n(\lambda_{k-1},\lambda_{\sigma(1)})\bigg)\nn\\
		&-\sum_{\substack{\sigma\in S_k\\\sigma(k)=k,\sigma(k-1)\neq k-1,\sigma(1)\neq k-1}}\prod_{i=1}^k\mathcal J_n(\lambda_{\sigma(i)},\lambda_{\sigma(i+1)}),
	\end{align}
	which is analytic along $\lambda_{k-1}=\lambda_k$.

	For $k=2$,~\eqref{hkk} reads
	\begin{align}
		&-\mathcal J_n^\circ(\lambda_1,\lambda_2)\mathcal J_n^\circ(\lambda_2,\lambda_1)-\frac{\mathcal J_n^\circ(\lambda_1,\lambda_2)-\mathcal J_n^\circ(\lambda_2,\lambda_1)}{\lambda_1-\lambda_2},
	\end{align}
	which is analytic along $\lambda_1=\lambda_2$. The lemma is proved.
	\end{proof}

	\begin{remark}
		By the three-term recurrence relations of $\pi_n$, we have
\begin{equation}
	L(\psi_A^\star(\lambda,n))=\lambda\psi_A^\star(\lambda,n),\qquad\psi_A^\star(\lambda,n)=(1+O(\lambda^{-1}))\lambda^n,\qquad n\in\ZZ_{\ge0}.
\end{equation}
By~\eqref{poincare} for $n\in\ZZ_{\ge0}$,
 \begin{equation}
	\psi_B^\star(\lambda,n)=(1+O(\lambda^{-1}))h_n\lambda^{-n}.
 \end{equation} Also for $n\in\ZZ_{\ge1}$,
 \begin{align}
	L(\psi_B^\star(\lambda,n))&=\psi_B^\star(\lambda,n+1)+\beta_n\psi_B^\star(\lambda,n)+\gamma_n\psi_B^\star(\lambda,n-1)\nn\\
	&\sim-\lambda\int_\RR\frac{1}{x-\lambda}(\pi_{n+1}(x)+\beta_n\pi_n(x)+\gamma_n\pi_{n-1}(x))\dif\mu(x)\nn\\
	&=-\lambda\int_\RR\frac{x\pi_n(x)}{x-\lambda}\dif\mu(x)\nn\\
 	&=-\lambda(\lambda\widehat\pi_n(\lambda)+h_0\delta_{n,0})\nn\\
	&\sim\lambda\psi_B^\star(\lambda,n).
 \end{align} In addition, from~\eqref{reduced-cd} for $n\in\ZZ_{\ge1}$, 
 \begin{align}
	\psi_A^\star(\lambda,n)\mathcal T^{-1}\psi_B^\star(\lambda,n)-\psi_B^\star(\lambda,n)\mathcal T^{-1}\psi_A^\star(\lambda,n)\sim&-\pi_n(\lambda)\lambda\widehat\pi_{n-1}(\lambda)+\lambda\widehat\pi_n(\lambda)\pi_{n-1}(\lambda)\nn\\
 	=&\lambda h_{n-1}. \label{wronskian}
 \end{align} 
 Thus for $n\in\ZZ_{\ge0}$, $\psi_A^\star(\lambda,n),\psi_B^\star(\lambda,n)$ satisfy the condition of forming a pair of wave functions.
	\end{remark}

\section{Proof of Theorem \ref{main}} \label{T1sec}

By~\eqref{asym-func} and Lemma~\ref{analytic}, we only need to prove for pairwise distinct $\lambda_1,\cdots,\lambda_k\in\CC\backslash\RR$,
\begin{align}
	\mathcal C_1(n;\lambda_1)=&\mathcal J_n^\circ(\lambda_1,\lambda_1)-\frac{n}{\lambda_1}, \label{mainanaly1}\\
	\mathcal C_k(n;\lambda_1,\cdots,\lambda_k)=&\frac{(-1)^{k-1}}{k}\sum_{\sigma\in S_k}\prod_{i=1}^k\mathcal J_n(\lambda_{\sigma(i)},\lambda_{\sigma(i+1)})-\frac{\delta_{k,2}}{(\lambda_1-\lambda_2)^2},\qquad k\ge2, \label{mainanaly}
\end{align}
(cf.~\cite{BE,DYZcp1,FY,Laguerre,JUE,corJUE}). The case $n=0$ is direct. We always let $n\ge1$ throughout this section. We first prove a lemma.

	\begin{lem}
		For any $l\ge k\ge 0$,
		\begin{align}
			\widehat{\pi_k\pi_l}(\lambda)=\pi_k(\lambda)\widehat\pi_l(\lambda).\label{ax1}
		\end{align}
	\end{lem}
	\begin{proof}
		$k=0$ is direct. For $k\ge1$, we have
		\begin{align}
			\widehat{\pi_k\pi_l}(\lambda)-\pi_k(\lambda)\widehat\pi_l(\lambda)=&\int_\RR\frac{\pi_k(y)-\pi_k(\lambda)}{y-\lambda}\pi_l(y)\dif\mu(y)=0.
		\end{align} The last equation follows from the fact that $\frac{\pi_k(y)-\pi_k(\lambda)}{y-\lambda}$ is a polynomial in $y$ of degree $k-1<l$. Thus the lemma is proved.
	\end{proof}
	
	The key method we use in the proof of Theorem \ref{main} will be based on the four fundamental properties~\eqref{c-d},~\eqref{trans-cd},~\eqref{reduced-cd} and~\eqref{ax1} of orthogonal polynomials. 

	\begin{proof}[Proof of Theorem~\ref{main}]
		We first prove~\eqref{mainanaly1}. Using~\eqref{ax1},
		\begin{align}
	\mathcal C_1(n;\lambda_1)=&\int_{\RR^n}\sum_{i=1}^n\frac{1}{\lambda_1-y_i}\rho_n(n;y_1,\cdots,y_n)\dif\mu(y_1)\cdots\dif\mu(y_n)-\frac{n}{\lambda_1}\nn\\
	=&\sum_{i=1}^n\int_{\RR}\frac{1}{\lambda_1-y_i}\rho_1(n;y_i)\dif\mu(y_i)-\frac{n}{\lambda_1}\nn\\
	=&\int_{\RR}\frac{1}{\lambda_1-y}\sum_{l=0}^{n-1}\frac{\pi_l(y)^2}{h_l}\dif\mu(y)-\frac{n}{\lambda_1}\nn\\
	=&-\sum_{l=0}^{n-1}\frac{\pi_l(\lambda_1)\widehat\pi_l(\lambda_1)}{h_l}-\frac{n}{\lambda_1}\nn\\
	=&\mathcal J_n^\circ(\lambda_1,\lambda_1)-\frac{n}{\lambda_1}. \label{pf1}
	\end{align}

Now we prove~\eqref{mainanaly}. For a partition $\mathcal P$ of $\mathbf Z_k$, we write $\mathcal P=\{\mathcal P_1,\cdots,\mathcal P_{|\mathcal P|}\}$.
	The following formula was given in~\cite{corJUE}: for
	$k\ge1$,
	\begin{equation} \label{ckbar}
		\left\langle\prod_{j=1}^k\tr\frac{1}{\lambda_j-M}\right\rangle(n)=\sum_{\text{partition }\mathcal P\text{ of }\mathbf Z_k}\int_{\RR^l}\prod_{i=1}^{|\mathcal P|}\left(\prod_{j \in \mathcal P_i} \frac{1}{\lambda_{j}-y_{i}}\right)\det(K_{n}(y_{i},y_{j}))_{i,j=1}^{l} \dif\mu(y_{1})\cdots\dif\mu(y_{|\mathcal P|}).
	\end{equation}
To prove~\eqref{ckbar}, we consider the following two cases. For the case when $n\ge k$,
	\begin{align}
		\left\langle\prod_{j=1}^k\tr\frac{1}{\lambda_j-M}\right\rangle(n)=&\int_{\RR^{n}}\prod_{j=1}^{k}\left(\sum_{i=1}^{n}\frac{1}{\lambda_{j}-y_{i}}\right)\rho_n(n;y_1,\cdots,y_n)\dif\mu(y_{1})\cdots\dif\mu(y_{n}) \nn \\
		=&\sum_{\text{partition }\mathcal P\text{ of }\mathbf Z_k}\int_{\RR^l}\prod_{i=1}^{|\mathcal P|}\left(\prod_{j \in \mathcal P_i} \frac{1}{\lambda_{j}-y_{i}}\right)\det(K_{n}(y_{i},y_{j}))_{i,j=1}^{l} \dif\mu(y_{1})\cdots\dif\mu(y_{|\mathcal P|}).
	\end{align}
	For the case when $n<k$,
	\begin{align}
		&\left\langle\prod_{j=1}^k\tr\frac{1}{\lambda_j-M}\right\rangle(n)\nn\\
		=&\sum_{\substack{\text{partition }\mathcal P\text{ of }\mathbf Z_k\\|\mathcal P|\le n}}\int_{\RR^l}\prod_{i=1}^{|\mathcal P|}\left(\prod_{j \in \mathcal P_i} \frac{1}{\lambda_{j}-y_{i}}\right)\det(K_{n}(y_{i},y_{j}))_{i,j=1}^{l} \dif\mu(y_{1})\cdots\dif\mu(y_{|\mathcal P|})\nn\\
		=&\left(\sum_{\text{partition }\mathcal P\text{ of }\mathbf Z_k}-\sum_{\substack{\text{partition }\mathcal P\text{ of }\mathbf Z_k\\|\mathcal P|>n}}\right)\int_{\RR^l}\prod_{i=1}^{|\mathcal P|}\left(\prod_{j \in \mathcal P_i} \frac{1}{\lambda_{j}-y_{i}}\right)\det(K_{n}(y_{i},y_{j}))_{i,j=1}^{|\mathcal P|} \dif\mu(y_{1})\cdots\dif\mu(y_{|\mathcal P|}).
	\end{align}
	Note that for any $m>n$, the matrix
	\begin{align}
		(K_n(y_i,y_j))_{i,j=1}^m=\left(\frac{\pi_l(y_i)}{\sqrt{h_l}}\right)_{i=1,\cdots,m;l=0,\cdots,n-1} \left(\frac{\pi_l(y_j)}{\sqrt{h_l}}\right)^T_{j=1,\cdots,m;l=0,\cdots,n-1}
	\end{align} has rank $\le n<m$. Thus $\det(K_n(y_i,y_j))_{i,j=1}^m=0$. Therefore we obtain~\eqref{ckbar}.

	Now we claim that for $k\ge1$,
	\begin{equation}
		\mathcal C_{k}(n;\lambda_{1},\cdots,\lambda_{k})=\sum_{\text{partition }\Pi\text{ of }\mathbf Z_k}C_{\mathbf Z_k,\Pi}-\frac{n\delta_{k,1}}{\lambda_1}, \label{ck'}
	\end{equation}
	where for $I\subset\mathbf Z_k$ and a partition $\mathcal P=\{P_1,\cdots,P_{|\mathcal P|}\}$ of $I$, we denote
	\begin{equation}
		C_{I,\mathcal P}:=\int_{\RR^{|\mathcal P|}}\prod_{i=1}^{|\mathcal P|}\left(\frac{\dif\mu(y_i)}{\prod_{j\in P_i}(\lambda_j-y_i)}\right)\frac{(-1)^{|\mathcal P|-1}}{|\mathcal P|}\sum_{\sigma\in S_{|\mathcal P|}}\prod_{i=1}^{|\mathcal P|}K_{n}(y_{\sigma(i)},y_{\sigma(i+1)}).
	\end{equation}
	The case $k=1$ is straightforward by definition. We prove by induction on $k$. Assume that~\eqref{ck'} holds for all $m<k$. 
	For $I\subset\mathbf Z_k$ and a partition $\mathcal P$ of $\mathbf Z_k$, we write $\mathcal P(I)$ for the partition of $I$ induced by $\mathcal P$. For two paritions $\Pi,\mathcal P$ of $\mathbf Z_k$, we write $\Pi\prec\mathcal P$ if $\Pi$ is finer than $\mathcal P$. 
	Then for case $k$,
	\begin{align}
		&\mathcal C_{k}(n;\lambda_{1},\cdots,\lambda_{k})\nn\\
		=&\left\langle\prod_{j=1}^k\tr\frac{1}{\lambda_j-M}\right\rangle(n)-\sum_{\substack{\text{partition }\mathcal{P}\text{ of }\mathbf Z_k\\|\mathcal{P}|>1}}\prod_{\substack{I\in\mathcal{P}\\I:=\{i_1,\cdots,i_{|I|}\}}}\left(\mathcal C_{|I|}(n;\lambda_{i_1},\cdots,\lambda_{i_{|I|}})+\frac{n\delta_{|I|,1}}{\lambda_{i_1}}\right) \nn \\
		=&\sum_{\text{partition }\Pi\text{ of }\mathbf Z_k}\int_{\RR^l}\prod_{i=1}^{|\mathcal P|}\left(\prod_{j \in \mathcal P_i} \frac{1}{\lambda_{j}-y_{i}}\right)\det(K_{n}(y_{i},y_{j}))_{i,j=1}^{l} \dif\mu(y_{1})\cdots\dif\mu(y_{|\mathcal P|})\nn\\
		&-\sum_{\substack{\text{partition }\mathcal{P}\text{ of }\mathbf Z_k\\|\mathcal{P}|>1}}\prod_{\substack{I\in\mathcal{P}\\I:=\{i_1,\cdots,i_{|I|}\}}}\left(\mathcal C_{|I|}(n;\lambda_{i_1},\cdots,\lambda_{i_{|I|}})+\frac{n\delta_{|I|,1}}{\lambda_{i_1}}\right)\nn \\
		=&\sum_{\text{partition }\Pi\text{ of }\mathbf Z_k}\sum_{\substack{\text{partition }\mathcal P\text{ of }\mathbf Z_k \\ \Pi\prec\mathcal P}}\prod_{I\in\mathcal P}C_{I,\Pi(I)}-\sum_{\substack{\text{partition }\mathcal P\text{ of }\mathbf Z_k \\ |\mathcal P|>1}}\prod_{I\in\mathcal P}\sum_{\text{partition }\Pi(I)\text{ of }I}C_{I,\Pi(I)} \nn\\
		=&\sum_{\text{partition }\Pi\text{ of }\mathbf Z_k}C_{\mathbf Z_k,\Pi}. \label{n>k}
	\end{align}

Now by~\eqref{ck'}, 
\begin{align}
	&\mathcal C_{k}(n;\lambda_{1},\cdots,\lambda_{k})\nn\\
	=&\sum_{\text{partition }\Pi\text{ of }\mathbf Z_k}\int_{\RR^{|\Pi|}} \prod_{i=1}^{|\Pi|}\left(\frac{\dif\mu(y_i)}{\prod_{j\in\Pi_i}(\lambda_j-y_i)}\right)\frac{(-1)^{|\Pi|-1}}{|\Pi|}\sum_{\sigma\in S_{|\Pi|}}\prod_{i=1}^{|\Pi|}K_n(y_{\sigma(i)},y_{\sigma(i+1)}) \nn \\
	=&-\sum_{\text{partition }\Pi\text{ of }\mathbf Z_k}\sum_{(a_1,\cdots,a_{|\Pi|})\in\Pi_1\times\cdots\times\Pi_{|\Pi|}} \left(\prod_{i=1}^{|\Pi|} Q_{a_i,\Pi_i\backslash\{a_i\}}(\lambda_1,\cdots,\lambda_k)\right)\frac{1}{|\Pi|}\sum_{\sigma\in S_{|\Pi|}}\ell_{|\Pi|}(n;\lambda_{a_{\sigma(1)}},\cdots,\lambda_{a_{\sigma(|\Pi|)}}),
\end{align}
where we denote the following rational functions in $\lambda_1,\cdots,\lambda_k$ for $I\subset \mathbf Z_k$ and $a\in \mathbf Z_k\backslash I$:
\begin{equation}
	Q_{a,I}(\lambda_1,\cdots,\lambda_k):=\prod_{b\in I}\frac{1}{\lambda_b-\lambda_a},
\end{equation}
and denote the cyclic integrals
\begin{align}
	\ell_0:=&1,\\
	\ell_k(n;\lambda_1,\cdots,\lambda_k):=&\int_{\RR^k}\left(\prod_{i=1}^{k}\frac{K_n(y_i,y_{i+1})}{y_i-\lambda_i}\dif\mu(y_i)\right),\ k\ge1,
\end{align} where $y_{k+1}:=y_1$.
To compute $\ell_k$, we first define
\begin{align}
	E_1(n;\lambda_1,\lambda_2):=&K_n(\lambda_1,\lambda_2),\\
	E_k(n;\lambda_1,\cdots,\lambda_{k+1}):=&\int_{\RR^{k-1}}K_n(\lambda_1,y_2)\left(\prod_{i=2}^{k-1}\frac{K_n(y_i,y_{i+1})}{y_i-\lambda_i}\dif\mu(y_i)\right)\frac{K_n(y_k,\lambda_{k+1})}{y_k-\lambda_k}\dif y_k,\qquad k\ge2.
\end{align}
Then
\begin{align}
	E_k(n;\lambda_1,\cdots,\lambda_{k+1})=&\int_\RR\frac{1}{y_k-\lambda_k}E_{k-1}(n;\lambda_1,\cdots,\lambda_{k-1},y_k)K_n(y_k,\lambda_{k+1})\dif\mu(y_k),\\
	\ell_k(n;\lambda_1,\cdots,\lambda_k)=&\int_\RR\frac{1}{y-\lambda_1}E_k(n;y,\lambda_2,\cdots,\lambda_k,y)\dif\mu(y).
\end{align}
\begin{lem}
	For $k\ge2$,
	\begin{equation} \label{ek}
		E_k(n;\lambda_1,\cdots,\lambda_{k+1})=-E_{k-1}(n;\lambda_1,\cdots,\lambda_k)\mathcal J_n(\lambda_{k+1},\lambda_k)-\frac{1}{\lambda_k-\lambda_{k+1}}E_{k-1}(n;\lambda_1,\cdots,\lambda_{k-1},\lambda_{k+1}).
	\end{equation}
\end{lem}

\begin{proof}
	For $k=2$, using~\eqref{c-d},~\eqref{trans-cd},~\eqref{reduced-cd} and~\eqref{ax1},
	\begin{align}
		&E_2(n;\lambda_1,\lambda_2,\lambda_3)+K_n(\lambda_1,\lambda_2)\mathcal J_n^\circ(\lambda_3,\lambda_2)\nn\\
		=&\sum_{l,m=0}^{n-1}\frac{1}{h_lh_m}(\pi_l(\lambda_1)\widehat{\pi_l\pi_m}(\lambda_2)\pi_m(\lambda_3)-\pi_l(\lambda_1)\pi_l(\lambda_2)\widehat\pi_m(\lambda_2)\pi_m(\lambda_3))\nn\\
		=&\sum_{l>m}\frac{1}{h_lh_m}\pi_l(\lambda_1)\pi_m(\lambda_3)(\widehat\pi_l(\lambda_2)\pi_m(\lambda_2)-\pi_l(\lambda_2)\widehat\pi_m(\lambda_2))\nn\\
		=&\sum_{l=1}^{n-1}\frac{\pi_l(\lambda_1)\widehat\pi_l(\lambda_2)}{h_lh_{l-1}}\frac{\pi_l(\lambda_2)\pi_{l-1}(\lambda_3)-\pi_{l-1}(\lambda_2)\pi_l(\lambda_3)}{\lambda_2-\lambda_3}\nn\\
		&-\sum_{l=1}^{n-1}\frac{\pi_l(\lambda_1)\pi_l(\lambda_2)}{h_l}\left(\frac{1}{h_{l-1}}\frac{\pi_l(\lambda_3)\widehat\pi_{l-1}(\lambda_2)-\pi_{l-1}(\lambda_3)\widehat\pi_l(\lambda_2)}{\lambda_3-\lambda_2} -\frac{1}{\lambda_2-\lambda_3}\right) \nn\\
		=&\sum_{l=1}^{n-1}\frac{\pi_l(\lambda_1)\pi_l(\lambda_3)}{(\lambda_2-\lambda_3)h_l}\frac{1}{h_{l-1}}(-\widehat\pi_l(\lambda_2)\pi_{l-1}(\lambda_2)+\widehat\pi_{l-1}(\lambda_2)\pi_l(\lambda_2))+\sum_{l=1}^{n-1}\frac{\pi_l(\lambda_1)\pi_l(\lambda_2)}{(\lambda_2-\lambda_3)h_l} \nn\\
		=&\frac{1}{\lambda_3-\lambda_2}(K_n(\lambda_1,\lambda_3)-K_n(\lambda_1,\lambda_2)). \label{method}
	\end{align} 
	Similarly,
	\begin{equation} \label{dk}
		\int_\RR\frac{1}{y_2-\lambda_2}\mathcal J_n(y_2,\lambda_1)K_n(y_2,\lambda_3)\dif\mu(y_2)=-\mathcal J_n(\lambda_2,\lambda_1)\mathcal J_n(\lambda_3,\lambda_2)-\left(\frac{1}{\lambda_1-\lambda_2}+\frac{1}{\lambda_2-\lambda_3}\right)\mathcal J_n(\lambda_3,\lambda_1).
	\end{equation}
	We prove~\eqref{ek} by induction on $k$. From now on, we will omit the parameter $n$ in $E_k$. Assume~\eqref{ek} holds for $k$. Now for $k+1$, using~\eqref{dk},
	\begin{align}
		&E_{k+1}(\lambda_1,\cdots,\lambda_{k+2})\nn\\
		=&E_{k-1}(\lambda_1,\cdots,\lambda_k)\left(\mathcal J_n(\lambda_{k+1},\lambda_k)\mathcal J_n(\lambda_{k+2},\lambda_{k+1})+\left(\frac{1}{\lambda_k-\lambda_{k+1}}+\frac{1}{\lambda_{k+1}-\lambda_{k+2}}\right)\mathcal J_n(\lambda_{k+2},\lambda_k)\right)\nn\\
		&+\frac{1}{\lambda_k-\lambda_{k+1}}(E_k(\lambda_1,\cdots,\lambda_k,\lambda_{k+2})-E_k(\lambda_1,\cdots,\lambda_{k-1},\lambda_{k+1},\lambda_{k+2}))\nn\\
		=&-\left(E_k(\lambda_1,\cdots,\lambda_{k+1})+\frac{1}{\lambda_k-\lambda_{k+1}}E_{k-1}(\lambda_1,\cdots,\lambda_{k-1},\lambda_{k+1})\right)\mathcal J_n(\lambda_{k+2},\lambda_{k+1})\nn\\
		&+\frac{1}{\lambda_k-\lambda_{k+1}}E_{k-1}(\lambda_1,\cdots,\lambda_k)\mathcal J_n(\lambda_{k+2},\lambda_k)\nn\\
		&-\frac{1}{\lambda_{k+1}-\lambda_{k+2}}\left(E_k(\lambda_1,\cdots,\lambda_k,\lambda_{k+2})+\frac{1}{\lambda_k-\lambda_{k+2}}E_{k-1}(\lambda_1,\cdots,\lambda_{k-1},\lambda_{k+2})\right)\nn\\
		&-\frac{1}{\lambda_k-\lambda_{k+1}}\left(E_{k-1}(\lambda_1,\cdots,\lambda_k)\mathcal J_n(\lambda_{k+2},\lambda_k)+\frac{1}{\lambda_k-\lambda_{k+2}}E_{k-1}(\lambda_1,\cdots,\lambda_{k-1},\lambda_{k+2})\right)\nn\\
		&+\frac{1}{\lambda_k-\lambda_{k+1}}\left(E_{k-1}(\lambda_1,\cdots,\lambda_{k-1},\lambda_{k+1})\mathcal J_n(\lambda_{k+2},\lambda_{k+1})+\frac{1}{\lambda_{k+1}-\lambda_{k+2}}E_{k-1}(\lambda_1,\cdots,\lambda_{k-1},\lambda_{k+2})\right)\nn\\
		=&-E_k(\lambda_1,\cdots,\lambda_{k+1})\mathcal J_n(\lambda_{k+2},\lambda_{k+1})-\frac{1}{\lambda_{k+1}-\lambda_{k+2}}E_k(\lambda_1,\cdots,\lambda_k,\lambda_{k+2}).
	\end{align}
	The lemma is proved.
\end{proof}

Having established a recurrence for $E_k$, we now turn to the cyclic integral $\ell_k$. To solve this recursively, we introduce an auxiliary double sequence of functions, $\ell_{k,(m)}$, defined as follows for $k\ge2$ and $1\le m\le k-2$:
\begin{align}
	\ell_{1,(1)}(n;\lambda_1):=&\ell_1(n;\lambda_1)=-\mathcal J_n^\circ(\lambda_1,\lambda_1),\\
	\ell_{k,(1)}(n;\lambda_1,\cdots,\lambda_k):=&\ell_k(n;\lambda_1,\cdots,\lambda_k)+\frac{1}{\lambda_k-\lambda_1}(\ell_{k-1}(n;\lambda_1,\cdots,\lambda_{k-1})-\ell_{k-1}(n;\lambda_2,\cdots,\lambda_k)),\\
	\ell_{k,(m+1)}(n;\lambda_1,\cdots,\lambda_k):=&\ell_{k,(m)}(n;\lambda_1,\cdots,\lambda_k)\nn\\
	&+\frac{1}{\lambda_{k-m}-\lambda_{k-m+1}}\ell_{k-1,(m)}(n;\lambda_1,\cdots,\lambda_{k-m-1},\lambda_{k-m+1},\cdots,\lambda_k),\\
	\ell_{k,(k)}(n;\lambda_1,\cdots,\lambda_k):=&\ell_{k,(k-1)}(n;\lambda_1,\cdots,\lambda_k)+\left(\frac{1}{\lambda_1-\lambda_2}+\frac{1}{\lambda_k-\lambda_1}\right)\ell_{k-1,(k-1)}(n;\lambda_2,\cdots,\lambda_k).
\end{align}
This construction allows us to relate the desired term $\ell_k$ to a more tractable term $\ell_{k,(k)}$ which we can compute directly.

\begin{lem}
	For $k\ge2$,
	\begin{equation}
		\ell_{k,(k)}(n;\lambda_1,\cdots,\lambda_k)=(-1)^k\prod_{i=1}^k\mathcal J_n(\lambda_{i+1},\lambda_i)-\frac{1}{\prod_{i=1}^k(\lambda_i-\lambda_{i+1})}. \label{lkk}
	\end{equation}
\end{lem}
\begin{proof}
	Throughout this proof we will omit the parameter $n$ in $\ell_{k,(m)}$. We first claim that for $1\le m\le k-1$,
	\begin{equation}
		\ell_{k,(m)}(\lambda_1,\cdots,\lambda_k)=\int_\RR\frac{(-1)^m}{y-\lambda_1}E_{k-m}(y,\lambda_2,\cdots,\lambda_{k-m+1})\left(\prod_{i=k-m+1}^{k-1}\mathcal J_n(\lambda_{i+1},\lambda_i)\right)\mathcal J_n(y,\lambda_k)\dif\mu(y). \label{claim}
	\end{equation} We prove this claim by induction on $m$. For $m=1$,
	\begin{align}
		&\ell_k(\lambda_1,\cdots,\lambda_k)\nn\\
		=&\int_\RR\frac{1}{y-\lambda_1}\left(-E_{k-1}(y,\lambda_2,\cdots,\lambda_k)\mathcal J_n(y,\lambda_k)-\frac{1}{\lambda_k-y}E_{k-1}(y,\lambda_2,\cdots,\lambda_{k-1},y)\right)\dif\mu(y)\nn\\
		=&-\int_\RR\frac{1}{y-\lambda_1}E_{k-1}(y,\lambda_2,\cdots,\lambda_k)\mathcal J_n(y,\lambda_k)\dif\mu(y)+\frac{1}{\lambda_1-\lambda_k}(\ell_{k-1}(\lambda_1,\cdots,\lambda_{k-1})-\ell_{k-1}(\lambda_2,\cdots,\lambda_k)),
	\end{align} 
	where in the last equation we use the fact that by definition $\ell_{k-1}(\lambda_k,\lambda_2,\cdots,\lambda_{k-1})=\ell_{k-1}(\lambda_2,\cdots,\lambda_k)$. Now assume that~\eqref{claim} holds for $m$. Then for $m+1$,
	\begin{align}
		&\ell_{k,(m)}(\lambda_1,\cdots,\lambda_k)\nn\\
		=&\int_\RR\frac{(-1)^m}{y-\lambda_1}E_{k-m}(y,\lambda_2,\cdots,\lambda_{k-m+1})\left(\prod_{i=k-m+1}^{k-1}\mathcal J_n(\lambda_{i+1},\lambda_i)\right)\mathcal J_n(y,\lambda_k)\dif\mu(y)\nn\\
		=&(-1)^m\left(\prod_{i=k-m+1}^{k-1}\mathcal J_n(\lambda_{i+1},\lambda_i)\right)\int_\RR\frac{\mathcal J_n(y,\lambda_k)}{y-\lambda_1}\bigg(-E_{k-m-1}(y,\lambda_2,\cdots,\lambda_{k-m})\mathcal J_n(\lambda_{k-m+1},\lambda_{k-m})\nn\\
		&\qquad\qquad\qquad\qquad\qquad\qquad\qquad\qquad-\frac{1}{\lambda_{k-m}-\lambda_{k-m+1}}E_{k-m-1}(y,\lambda_2,\cdots,\lambda_{k-m-1},\lambda_{k-m+1})\bigg)\dif\mu(y)\nn\\
		=&\int_\RR\frac{(-1)^{m+1}}{y-\lambda_1}E_{k-m-1}(y,\lambda_2,\cdots,\lambda_{k-m})\left(\prod_{i=k-m}^{k-1}\mathcal J_n(\lambda_{i+1},\lambda_i)\right)\mathcal J_n(y,\lambda_k)\dif\mu(y)\nn\\
		&-\frac{1}{\lambda_{k-m}-\lambda_{k-m+1}}\ell_{k-1,(m)}(\lambda_1,\cdots,\lambda_{k-m-1},\lambda_{k-m+1},\cdots,\lambda_k).
	\end{align} Thus we have proved~\eqref{claim}.
	
	Now we prove~\eqref{lkk} by induction on $k$. For $k=2$, using an argument similar to that of~\eqref{method},
	\begin{align}
		\ell_{2,(1)}(\lambda_1,\lambda_2)=&-\int_\RR\frac{1}{y-\lambda_1}E_1(y,\lambda_2)\mathcal J_n(y,\lambda_2)\dif\mu(y)\nn\\
		=&\mathcal J_n(\lambda_2,\lambda_1)\mathcal J_n(\lambda_1,\lambda_2)+\frac{1}{(\lambda_1-\lambda_2)^2}.
	\end{align}
	Then
	\begin{align}
		\ell_{2,(2)}(\lambda_1,\lambda_2)=&\ell_{2,(1)}(\lambda_1,\lambda_2)+\left(\frac{1}{\lambda_1-\lambda_2}+\frac{1}{\lambda_2-\lambda_1}\right)\ell_{1,(1)}(\lambda_2)\nn\\
		=&\mathcal J_n(\lambda_2,\lambda_1)\mathcal J_n(\lambda_1,\lambda_2)-\frac{1}{(\lambda_1-\lambda_2)(\lambda_2-\lambda_1)}.
	\end{align} 
	Assume that~\eqref{lkk} holds for $k-1$, then for $k$, following~\eqref{dk},
	\begin{align}
		&\ell_{k,(k-1)}(\lambda_1,\cdots,\lambda_k)\nn\\
		=&(-1)^{k-1}\left(\prod_{i=2}^{k-1}\mathcal J_n(\lambda_{i+1},\lambda_i)\right)\int_\RR\frac{1}{y-\lambda_1}\mathcal J_n(y,\lambda_k)K_n(y,\lambda_2)\dif\mu(y)\nn\\
		=&(-1)^k\left(\prod_{i=2}^{k-1}\mathcal J_n(\lambda_{i+1},\lambda_i)\right)\left(\mathcal J_n(\lambda_1,\lambda_k)\mathcal J_n(\lambda_2,\lambda_1)+\left(\frac{1}{\lambda_k-\lambda_1}+\frac{1}{\lambda_1-\lambda_2}\right)\mathcal J_n(\lambda_2,\lambda_k)\right)\nn\\
		=&(-1)^k\left(\prod_{i=1}^k\mathcal J_n(\lambda_{i+1},\lambda_i)+\left(\frac{1}{\lambda_1-\lambda_2}+\frac{1}{\lambda_k-\lambda_1}\right)\mathcal J_n(\lambda_2,\lambda_k)\prod_{i=2}^{k-1}\mathcal J_n(\lambda_{i+1},\lambda_i)\right).
	\end{align} 
	Thus
	\begin{align}
		&\ell_{k,(k)}(\lambda_1,\cdots,\lambda_k)\nn\\
		=&\ell_{k,(k-1)}(\lambda_1,\cdots,\lambda_k)+\left(\frac{1}{\lambda_1-\lambda_2}+\frac{1}{\lambda_k-\lambda_1}\right)\ell_{k-1,(k-1)}(\lambda_2,\cdots,\lambda_k)\nn\\
		=&(-1)^k\left(\prod_{i=1}^k\mathcal J_n(\lambda_{i+1},\lambda_i)+\left(\frac{1}{\lambda_1-\lambda_2}+\frac{1}{\lambda_k-\lambda_1}\right)\mathcal J_n(\lambda_2,\lambda_k)\prod_{i=2}^{k-1}\mathcal J_n(\lambda_{i+1},\lambda_i)\right)\nn\\
		&+\left(\frac{1}{\lambda_1-\lambda_2}+\frac{1}{\lambda_k-\lambda_1}\right)\left((-1)^{k-1}\mathcal J_n(\lambda_2,\lambda_k)\prod_{i=2}^{k-1}\mathcal J_n(\lambda_{i+1},\lambda_i)-\frac{1}{(\lambda_k-\lambda_2)\prod_{i=1}^{k-1}(\lambda_i-\lambda_{i+1})}\right)\nn\\
		=&(-1)^k\prod_{i=1}^k\mathcal J_n(\lambda_{i+1},\lambda_i)-\frac{1}{\prod_{i=1}^k(\lambda_i-\lambda_{i+1})}.
	\end{align}
	The lemma is proved.
\end{proof}

	The double sequence $\{\ell_{k,(m)}\}$ forms the following table:
	\begin{equation} \label{table}
		\begin{array}{c|cccccc}
			\text{row}\backslash\text{column} & 0 & 1 & 2 & \dots & k-1 & k \\
			\hline
			0 & \ell_0 & \ell_1 & \ell_2 & \dots & \ell_{k-1} & \ell_k \\
			1 & & \ell_{1,(1)} & \ell_{2,(1)} & \dots & \ell_{k-1,(1)} & \ell_{k,(1)} \\
			2 & & & \ell_{2,(2)} & \dots & \ell_{k-1,(2)} & \ell_{k,(2)} \\
			\vdots & & & & \ddots & \vdots & \vdots \\
			k-1 & & & & & \ell_{k-1,(k-1)} & \ell_{k,(k-1)} \\
			k & & & & & & \ell_{k,(k)}
		\end{array}
	\end{equation} 
	Each entry in this table is a sum of the terms located directly above and to the upper-left. Given the diagonal entries, we can uniquely determine the entire table, specifically the 0th-row entries.

Denote the following rational functions in $\lambda_1,\cdots,\lambda_k$ for $\{a_1,\cdots,a_t\}\subset \mathbf Z_k$:
\begin{equation}
	R_{a_1,\cdots,a_t}(\lambda_1,\cdots,\lambda_k):=\prod_{b\in \mathbf Z_k\backslash\{a_1,\cdots,a_t\}}\frac{1}{\lambda_b-\lambda_{b+1}}.
\end{equation}

\begin{lem}
	For $k\ge2$,
	\begin{align}
		\ell_k(n;\lambda_1,\cdots,\lambda_k)=&(-1)^k\prod_{i=1}^k\mathcal J_n(\lambda_{i+1},\lambda_i)-\prod_{i=1}^k\frac{1}{(\lambda_i-\lambda_{i+1})}\nn\\
		&-\sum_{t=1}^{k-1}\sum_{1\le a_1<\cdots<a_t\le k}R_{a_1,\cdots,a_t}(\lambda_1,\cdots,\lambda_k)\ell_t(n;\lambda_{a_1},\cdots,\lambda_{a_t}). \label{lk}
	\end{align}
\end{lem}
\begin{proof}
	We first outline the strategy of the proof. In the context of Table~\eqref{table}, equation~\eqref{lk} can be interpreted as expressing the entry at $(0,k)$ in terms of the entries at $(0,m)$, $0\le m\le k-1$ and the diagonal entry at $(k,k)$. To achieve this, we initially express the entry at $(0,k)$ in terms of the entries at the $(k-1)$th column and at the diagonal term $(k,k)$. Subsequently, in each step, we iteratively reduce the dependence on the $m$-th column to the $(m-1)$-th column and the boundary entry at $(0,m)$. This procedure ultimately yields~\eqref{lk}. The process is illustrated as follows:
	\[
	\begin{array}{c|c}
		 & k \\ \hline
		 0 & *
	\end{array}
	\Rightarrow
	\begin{array}{c}
		\begin{array}{c|cc}
			 & k-1 & k \\ \hline
			0 & * & * \\
			1 & * & 0 \\
			\vdots & \vdots & \vdots \\
			k-1 & * & 0 \\
			k & & * 
		\end{array}
	\end{array}
	\Rightarrow
	\begin{array}{c}
		\begin{array}{c|ccc}
			& k-2 & k-1 & k \\ \hline
			0 & * & * & * \\
			1 & * & 0 & 0 \\
			\vdots & \vdots & \vdots & \vdots \\
			k-2 & * & 0 & 0 \\
			k-1 & & 0 & 0 \\
			k & &  & * 
		\end{array}
	\end{array}
	\Rightarrow \dots \Rightarrow
	\begin{array}{c}
		\begin{array}{c|ccccc}
			& 0 & 1 & \dots & k-1 & k \\ \hline
			0 & * & * & \dots & * & * \\
			1 & & 0 & \dots & 0 & 0 \\
			\vdots & &  & \ddots & \vdots & \vdots \\
			k-1 & &  &  & 0 & 0 \\ 
			k & & &  & & * 
		\end{array}
	\end{array}
	\]
	
	More precisely, for $m\le k-1$ we prove the following formula, which expresses $\ell_k$ in terms of the entries at $(0,r)$ for $k-m\le r\le k-1$, the entries at $(r,k-m)$ for $1\le r\le k-m$, and the diagonal term $(k,k)$: (We still omit the parameter $n$.)
	\begin{align}
		&\ell_k(\lambda_1,\cdots,\lambda_k)\nn \\
		=& \ell_{k,(k)}-\sum_{r=1}^{m-1}\sum_{1\le a_1<\cdots<a_r\le k}\widetilde R_{a_1,\cdots,a_r}\ell_{k-r}^{\widehat{a_1,\cdots,a_r}} \nn\\
		& -\sum_{r=0}^{m-1}\sum_{r+1<a_{r+1}<\cdots<a_{m-1}\le k-1}\widetilde R_{1,\cdots,r,a_{r+1},\cdots,a_{m-1}}\frac{\ell_{k-m}^{\widehat{1,\cdots,r,a_{r+1},\cdots,a_{m-1},k}}-\ell_{k-m}^{\widehat{1,\cdots,r,r+1,a_{r+1},\cdots,a_{m-1}}}}{\lambda_k-\lambda_1} \nn\\
		& -\sum_{t=1}^{k-m-1}\sum_{r=0}^{m-1}\sum_{r+1<a_{r+1}<\cdots<a_m=k-t} \frac{\lambda_k-\lambda_{r+1}}{\lambda_k-\lambda_1}\widetilde R_{1,\cdots,r,a_{r+1},\cdots,a_m}\ell_{k-m,(t)}^{\widehat{1,\cdots,r,a_{r+1},\cdots,a_m}}\nn\\
		&-\frac{\lambda_k-\lambda_{m+1}}{\lambda_k-\lambda_1}\widetilde R_{1,\cdots,m}\ell_{k-m,(k-m)}^{\widehat{1,\cdots,m}}, \label{lind}
	\end{align}
	where for $a_{1}<\cdots<a_{m}$ we denote
	\begin{align}
		\widetilde R_{a_1,\cdots,a_m}:=& \prod_{s=1}^{m}\frac{1}{\lambda_{a_s}-\lambda_{a_s+1}}, \\
		\ell_{k-m}^{\widehat{a_1,\cdots,a_m}} :=& \ell_{k-m}(\lambda_1,\cdots,\lambda_{a_1-1},\lambda_{a_1+1},\cdots,\lambda_{a_m-1},\lambda_{a_m+1},\cdots,\lambda_k), \\
		\ell_{k-m,(t)}^{\widehat{a_1,\cdots,a_m}} :=& \ell_{k-m,(t)}(\lambda_1,\cdots,\lambda_{a_1-1},\lambda_{a_1+1},\cdots,\lambda_{a_m-1},\lambda_{a_m+1},\cdots,\lambda_k).
	\end{align}
	
	We prove~\eqref{lind} by induction on $m$. First, for $m=1$,
	\begin{align}
		\ell_{k}(\lambda_1,\cdots,\lambda_k) &= \ell_{k,(1)}-\frac{1}{\lambda_k-\lambda_1}(\ell_{k-1}^{\widehat k}-\ell_{k-1}^{\widehat1}) \nn\\
		&= \ell_{k,(2)}-\frac{1}{\lambda_k-\lambda_1}(\ell_{k-1}^{\widehat k}-\ell_{k-1}^{\widehat1})-\frac{1}{\lambda_{k-1}-\lambda_k}\ell_{k-1,(1)}^{\widehat{k-1}} \nn\\
		&= \ell_{k,(k-1)}-\frac{1}{\lambda_k-\lambda_1}(\ell_{k-1}^{\widehat k}-\ell_{k-1}^{\widehat1})-\sum_{t=1}^{k-2}\frac{1}{\lambda_{k-t}-\lambda_{k-t+1}}\ell_{k-1,(t)}^{\widehat{k-t}} \nn\\
		&= \ell_{k,(k)}-\frac{1}{\lambda_k-\lambda_1}(\ell_{k-1}^{\widehat k}-\ell_{k-1}^{\widehat1})-\sum_{t=1}^{k-2}\widetilde R_{k-t}\ell_{k-1,(t)}^{\widehat{k-t}} -\frac{\lambda_k-\lambda_2}{\lambda_k-\lambda_1}\widetilde R_1\ell_{k-1,(k-1)}^{\widehat1}.
	\end{align}
	
	Now assume that~\eqref{lind} holds for $m$. For $m+1$, we transition from the $(k-m)$-th column to the $(k-m-1)$-th column. Consequently, we need to compute the terms at the entries $(t,k-m-1)$ for $0\le t\le k-m-1$, as well as the boundary term $(0,k-m)$. The term at $(k-m-1,k-m-1)$ is derived from the term at $(k-m,k-m)$ and is given by
	\begin{align}
		&-\left(\frac{1}{\lambda_{m+1}-\lambda_{m+2}}+\frac{1}{\lambda_k-\lambda_{m+1}}\right)\frac{\lambda_k-\lambda_{m+1}}{\lambda_k-\lambda_1}\widetilde R_{1,\cdots,m}\ell_{k-m-1,(k-m-1)}^{\widehat{1,\cdots,m,m+1}}\nn\\
		=&-\frac{\lambda_k-\lambda_{m+2}}{\lambda_k-\lambda_1}\widetilde R_{1,\cdots,m+1}\ell_{k-m-1,(k-m-1)}^{\widehat{1,\cdots,m+1}}. \label{part1}
	\end{align}
	The term at $(t,k-m-1)$ for $1\le t\le k-m-2$ is
	\begin{align}
		&-\frac{1}{\lambda_{k-t}-\lambda_{k-t+1}}\sum_{r=0}^{m}\sum_{r+1<a_{r+1}<\cdots<a_m\le k-t-1}\frac{\lambda_k-\lambda_{r+1}}{\lambda_k-\lambda_1}\widetilde R_{1,\cdots,r,a_{r+1},\cdots,a_m}\ell_{k-m-1,(t)}^{\widehat{1,\cdots,r,a_{r+1},\cdots,a_m,k-t}}\nn\\
		=&-\sum_{r=0}^{m+1}\sum_{r+1<a_{r+1}<\cdots<a_{m+1}=k-t}\frac{\lambda_k-\lambda_{r+1}}{\lambda_k-\lambda_1}\widetilde R_{1,\cdots,r,a_{r+1},\cdots,a_{m+1}}\ell_{k-m-1,(t)}^{\widehat{1,\cdots,r,a_{r+1},\cdots,a_{m+1}}}. \label{part2}
	\end{align}
	Similarly, the term at $(0,k-m-1)$ reads
	\begin{align}
		&-\sum_{r=0}^{m}\sum_{r+1<a_{r+1}<\cdots<a_m\le k-1}\frac{\lambda_k-\lambda_{r+1}}{\lambda_k-\lambda_1}\widetilde R_{1,\cdots,r,a_{r+1},\cdots,a_m}\frac{\ell_{k-m-1}^{\widehat{1,\cdots,r,a_{r+1},\cdots,a_m,k}}-\ell_{k-m-1}^{\widehat{1,\cdots,r+1,a_{r+1},\cdots,a_m}}}{\lambda_k-\lambda_{r+1}}\nn\\
		=&-\sum_{r=0}^{m}\sum_{r+1<a_{r+1}<\cdots<a_m\le k-1}\widetilde R_{1,\cdots,r,a_{r+1},\cdots,a_m}\frac{\ell_{k-m-1}^{\widehat{1,\cdots,r,a_{r+1},\cdots,a_m,k}}-\ell_{k-m-1}^{\widehat{1,\cdots,r+1,a_{r+1},\cdots,a_m}}}{\lambda_k-\lambda_1}. \label{part3}
	\end{align}
	Finally, the term at $(0,k-m)$ changes to
	\begin{align}
		& -\sum_{r=0}^{m-1}\sum_{r+1<a_{r+1}<\cdots<a_{m-1}\le k-1}\widetilde R_{1,\cdots,r,a_{r+1},\cdots,a_{m-1}}\frac{\ell_{k-m}^{\widehat{1,\cdots,r,a_{r+1},\cdots,a_{m-1},k}}-\ell_{k-m}^{\widehat{1,\cdots,r+1,a_{r+1},\cdots,a_{m-1}}}}{\lambda_k-\lambda_1} \nn\\
		&\qquad -\sum_{t=1}^{k-m}\sum_{r=0}^{m}\sum_{r+1<a_{r+1}<\cdots<a_m=k-t} \frac{\lambda_k-\lambda_{r+1}}{\lambda_k-\lambda_1}\widetilde R_{1,\cdots,r,a_{r+1},\cdots,a_m}\ell_{k-m}^{\widehat{1,\cdots,r,a_{r+1},\cdots,a_m}}\nn\\
		&\qquad -\frac{\lambda_k-\lambda_{m+1}}{\lambda_k-\lambda_1}\widetilde R_{1,\cdots,m}\ell_{k-m}^{\widehat{1,\cdots,m}}.\label{0k-m}
	\end{align}
	In~\eqref{0k-m} we compute the rational function before $\ell_{k-m}^{\widehat{1,\cdots,r,a_{r+1},\cdots,a_{m}}}$ where $a_{r+1}>r+1$ and $a_{m}<k$. It is
	\begin{align}
		&-\sum_{t=0}^{r-1}\frac{1}{\lambda_1-\lambda_k}\widetilde R_{1,\cdots,t,t+2,\cdots,r,a_{r+1},\cdots,a_{m-1}}-\frac{\lambda_k-\lambda_{r+1}}{\lambda_k-\lambda_1}\widetilde R_{1,\cdots,r,a_{r+1},\cdots,a_m} \nn\\
		=&-\frac{1}{\lambda_1-\lambda_k}\left(\sum_{t=0}^{r-1}(\lambda_{t+1}-\lambda_{t+2})-(\lambda_k-\lambda_{r+1})\right)\widetilde R_{1,\cdots,r,a_{r+1},\cdots,a_{m}} \nn\\
		=&-\widetilde R_{1,\cdots,r,a_{r+1},\cdots,a_m}.
	\end{align}
	Similarly, the rational function before $\ell_{k-m}^{\widehat{1,\cdots,m}}$ is $-\widetilde R_{1,\cdots,m}$. And lastly, the rational function before $\ell_{k-m}^{\widehat{1,\cdots,r,a_{r+1},\cdots,a_{m-1},k}}$ where $a_{r+1}>r+1$ is
	\begin{align}
	-\frac{1}{\lambda_k-\lambda_1}\widetilde R_{1,\cdots,r,a_{r+1},\cdots,a_{m-1}} =-\widetilde R_{1,\cdots,r,a_{r+1},\cdots,a_{m-1},k}.
	\end{align}
	Hence~\eqref{0k-m} reads
	\begin{equation}
		-\sum_{1\le a_1<\cdots<a_m\le k}\widetilde R_{a_1,\cdots,a_m}\ell_{k-m}^{\widehat{a_1,\cdots,a_m}}. \label{part4}
	\end{equation}
	Thus by~\eqref{part1},~\eqref{part2},~\eqref{part3} and~\eqref{part4}, we obtain~\eqref{lind} for $m+1$.
	
	Now following~\eqref{lind}, for $m=k-1$ we have
	\begin{align}
		&\ell_k(\lambda_1,\cdots,\lambda_k)\nn\\
		=& \ell_{k,(k)}-\sum_{r=1}^{k-2}\sum_{1\le a_1<\cdots<a_r\le k}\widetilde R_{a_1,\cdots,a_r}\ell_{k-r}^{\widehat{a_1,\cdots,a_r}}-\sum_{r=0}^{k-2}\widetilde R_{1,\cdots,r,r+2,\cdots,k-1}\frac{\ell_1(\lambda_{r+1})-\ell_1(\lambda_k)}{\lambda_k-\lambda_1} \nn\\
		=& \ell_{k,(k)}-\sum_{r=1}^{k-1}\sum_{1\le a_1<\cdots<a_r\le k}\widetilde R_{a_1,\cdots,a_r}\ell_{k-r}^{\widehat{a_1,\cdots,a_r}} \nn\\
		=&(-1)^k\prod_{i=1}^k\mathcal J_n(\lambda_{i+1},\lambda_i)-\frac{1}{\prod_{i=1}^k(\lambda_i-\lambda_{i+1})}-\sum_{t=1}^{k-1}\sum_{1\le a_1<\cdots<a_t\le k}R_{a_1,\cdots,a_t}(\lambda_1,\cdots,\lambda_k)\ell_t(\lambda_{a_1},\cdots,\lambda_{a_t}).
	\end{align}
	Thus the lemma is proved.
\end{proof}

\begin{lem}
	For $k\ge2$,
	\begin{align}
		&\frac{1}{k}\sum_{\sigma\in S_k}\ell_k(n;\lambda_{\sigma(1)},\cdots,\lambda_{\sigma(k)})\nn\\
		=&\frac{(-1)^k}{k}\sum_{\sigma\in S_k}\prod_{i=1}^k\mathcal J_n(\lambda_{\sigma(i)},\lambda_{\sigma(i+1)})+\frac{\delta_{k,2}}{(\lambda_1-\lambda_2)^2}\nn\\
		&-\sum_{t=1}^{k-1}\sum_{1\le a_1<\cdots<a_t\le k}\left(\sum_{\sqcup_{s=1}^tP_s=\mathbf Z_k\backslash\{a_1,\cdots,a_t\}}\prod_{i=1}^tQ_{a_i,P_i}(\lambda_1,\cdots,\lambda_k)\right)\frac{1}{t}\sum_{\sigma\in S_t}\ell_t(n;\lambda_{a_{\sigma(1)}},\cdots,\lambda_{a_{\sigma(t)}}). \label{cor-L}
	\end{align}
\end{lem}
\begin{proof}
	By~\eqref{lk}, noting that $\ell_{k}$ and $\ell_{k,(k)}$ are invariant under the cyclic group $C_{k}$, we have
	\begin{align}
		&\sum_{\sigma\in S_k/C_k}\ell_k(\lambda_{\sigma(1)},\cdots,\lambda_{\sigma(k)})\nn\\
		=&(-1)^k\sum_{\sigma\in S_k/C_k}\prod_{i=1}^k\mathcal J_n(\lambda_{\sigma(i)},\lambda_{\sigma(i+1)})-\sum_{\sigma\in S_k/C_k}\prod_{i=1}^k\frac{1}{\lambda_{\sigma(i)}-\lambda_{\sigma(i+1)}} \nn\\
		&-\sum_{\sigma\in S_{k}/C_{k}}\sum_{t=1}^{k-1}\sum_{1\le a_1<\cdots<a_t\le k}R_{\sigma(a_1),\cdots,\sigma(a_t)}(\lambda_{\sigma(1)},\cdots,\lambda_{\sigma(k)})\ell_{t}(\lambda_{\sigma(a_1)},\cdots,\lambda_{\sigma(a_t)}).
	\end{align}
	For $k\ge3$,
	\begin{align}
		\sum_{\sigma\in S_k/C_k}\prod_{i=1}^k\frac{1}{\lambda_{\sigma(i)}-\lambda_{\sigma(i+1)}}=\sum_{\sigma\in S_{k-1}}\left(\prod_{i=1}^{k-1}\frac{1}{\lambda_{\sigma(i)}-\lambda_{\sigma(i+1)}}\right)\left(\frac{1}{\lambda_{\sigma(k-1)}-\lambda_k}-\frac{1}{\lambda_{\sigma(1)}-\lambda_k}\right)=0.
	\end{align}
	Given $1\le a_{1}<\cdots<a_{t}\le k$, $\ell_{t}(\lambda_{a_{\sigma(1)}},\cdots,\lambda_{a_{\sigma(t)}})$ has the same rational function before it for any $\sigma\in S_{t}/C_{t}$, which is
	\begin{align}
		\sum_{\substack{\sqcup_{s=1}^tP_s=\mathbf Z_k\backslash\{a_1,\cdots,a_t\}\\P_s:=\{p_1,\cdots,p_{g_s}\}}}\prod_{s=1}^{t}\sum_{\sigma\in S_{g_s}}\left(\prod_{b=1}^{g_s-1}\frac{1}{\lambda_{p_{\sigma(b)}}-\lambda_{p_{\sigma(b+1)}}}\right)\frac{1}{\lambda_{p_{\sigma(g_s)}}-\lambda_{a_s}}.
	\end{align}
	Now we claim that for any set $P=\{p_{1},\cdots,p_{s}\}\subset\mathbf Z_k$ and $a\in\mathbf Z_k\backslash P$
	\begin{equation}
		\sum_{\sigma\in S_s}\left(\prod_{b=1}^{s-1}\frac{1}{\lambda_{p_{\sigma(b)}}-\lambda_{p_{\sigma(b+1)}}}\right)\frac{1}{\lambda_{p_{\sigma(s)}}-\lambda_{a}}=\prod_{b=1}^s\frac{1}{\lambda_{p_{b}}-\lambda_{a}},\label{claimQ}
	\end{equation}
	where the RHS is just $Q_{a,P}(\lambda_1,\cdots,\lambda_k)$. We prove~\eqref{claimQ} by induction on s. The base case $s=1$ is straightforward. Assume that~\eqref{claimQ} holds for $s\ge1$. Then for $s+1$, we consider the two sides of~\eqref{claimQ} as meromorphic functions in $\lambda_{a}$, treating the other $\lambda_{j}$'s as fixed parameters. Both rational functions possess only simple poles at $\lambda_{p_{1}},\cdots,\lambda_{p_{s+1}}$. Therefore, it suffices to verify that the residues at these poles match. Without loss of generality, we focus on the coefficient of the term $\frac{1}{\lambda_{p_{s+1}}-\lambda_{a}}$. Specifically, the coefficient on the LHS of~\eqref{claimQ} is
	\begin{align}
		\sum_{\substack{\sigma\in S_{s+1}\\\sigma(s+1)=s+1}}\prod_{b=1}^{s}\frac{1}{\lambda_{p_{\sigma(b)}}-\lambda_{p_{\sigma(b+1)}}}
		&=\sum_{\sigma\in S_s}\left(\prod_{b=1}^{s-1}\frac{1}{\lambda_{p_{\sigma(b)}}-\lambda_{p_{\sigma(b+1)}}}\right)\frac{1}{\lambda_{p_{\sigma(s)}}-\lambda_{p_{s+1}}}=\prod_{b=1}^s\frac{1}{\lambda_{p_b}-\lambda_{p_{s+1}}},
	\end{align}
	which is the same as the coefficient on the RHS of~\eqref{claimQ}. These prove the lemma.
\end{proof}

Now we are ready to finish the proof of Theorem~\ref{main}. For $k\ge2$,
	\begin{align}
		&C_{k}(n;\lambda_1,\cdots,\lambda_k)\nn\\
		=&-\sum_{\text{partition }\Pi\text{ of }\mathbf Z_k}\sum_{(a_1,\cdots,a_{|\Pi|})\in\Pi_1\times\cdots\times\Pi_{|\Pi|}} \left(\prod_{i=1}^{|\Pi|} Q_{a_i,\Pi_i\backslash\{a_i\}}(\lambda_1,\cdots,\lambda_k)\right)\frac{1}{|\Pi|}\sum_{\sigma\in S_{|\Pi|}} \ell_{|\Pi|}(n;\lambda_{a_{\sigma(1)}},\cdots,\lambda_{a_{\sigma(|\Pi|)}})\nn\\
		=&-\sum_{l=1}^k\sum_{1\le a_1<\cdots<a_l\le k} \sum_{\sqcup_{i=1}^l\Pi_i=\mathbf Z_k\setminus\{a_1,\cdots,a_l\}}\left(\prod_{i=1}^{l}Q_{a_i,\Pi_i}(\lambda_1,\cdots,\lambda_k)\right)\frac{1}{l}\sum_{\sigma\in S_l}\ell_l(n;\lambda_{a_{\sigma(1)}},\cdots,\lambda_{a_{\sigma(l)}})\nn\\
		=&\frac{(-1)^{k-1}}{k}\sum_{\sigma\in S_k}\prod_{i=1}^k\mathcal J_n(\lambda_{\sigma(i)},\lambda_{\sigma(i+1)})-\frac{\delta_{k,2}}{(\lambda_1-\lambda_2)^2},
	\end{align}
	where the last equation follows from~\eqref{cor-L}. Therefore we obtain the theorem.
\end{proof}	

\begin{remark}
	The matrix-resolvent method of the 1D Toda lattice hierarchy also applies to more general solutions, such as the one given by the Gromov-Witten invariants of $\mathbb{CP}^1$~\cite{DY,DYZcp1,Toda}.
\end{remark}

\begin{remark}
The formulas~\eqref{mainanaly1} and~\eqref{mainanaly} can be equivalently (cf.~\cite{BE,DY,DYZcp1,Toda,Ruzza,Laguerre,JUE,corJUE,LUE}) written as
\begin{align}
	\mathcal C_k(n;\lambda_1,\cdots,\lambda_k)=&-\sum_{\sigma\in S_k/C_k}\frac{\tr(R_n(\lambda_{\sigma(1)})\cdots R_n(\lambda_{\sigma(k)}))}{\prod_{i=1}^k(\lambda_{\sigma(i)}-\lambda_{\sigma(i+1)})}-\frac{\delta_{k,2}}{(\lambda_1-\lambda_2)^2},\qquad k\ge2,
\end{align}
(cf. also~\cite{BDY,BDYds}), where
\begin{align}
	R_n(\lambda):=&\frac{1}{h_{n-1}}\begin{pmatrix}
		-\pi_n(\lambda)\widehat\pi_{n-1}(\lambda)&-\pi_n(\lambda)\widehat\pi_n(\lambda)\\ 
		\pi_{n-1}(\lambda)\widehat\pi_{n-1}(\lambda)&\pi_{n-1}(\lambda)\widehat\pi_n(\lambda)
	\end{pmatrix}.
\end{align}
\end{remark}

\begin{remark}
For the cases of Jacobi unitary ensemble and Laguerre unitary ensemble, Gisonni, Grava and Ruzza~\cite{Laguerre,JUE,corJUE} also gave explicit formulas in terms of generalized hypergeometric functions using different methods (cf.~\cite{LUE}).
\end{remark}
	
\begin{remark}
	For the case of Gaussian unitary ensemble (GUE), the Christoffel--Darboux kernel reads
	\begin{equation}
		K_n^\text{GUE}(\lambda_1,\lambda_2)=\frac{1}{h_{n-1}}\frac{\He_n(\lambda_1)\He_{n-1}(\lambda_2)-\He_{n-1}(\lambda_1)\He_n(\lambda_2)}{\lambda_1-\lambda_2}.
	\end{equation} It is well known~\cite{PR,Szego,Forrester} that the double-scaling limit $n\to\infty$ of $K_n^\text{GUE}(\lambda_1,\lambda_2)$ gives the Airy kernel $K^\text{Airy}(\xi_1,\xi_2)$:
	\begin{equation}
		\lim_{n\to\infty}n^{-1/6}K_n^\mathrm{GUE}(\lambda_1,\lambda_2)e^{-\frac{\lambda_1^2+\lambda_2^2}{4}}=K^\mathrm{Airy}(\xi_1,\xi_2),
	\end{equation}
	where
	\begin{equation}
		\lambda_1=2\sqrt{n}+n^{-1/6}\xi_1,\qquad \lambda_2=2\sqrt{n}+n^{-1/6}\xi_2,
	\end{equation}
	and
	\begin{equation}
		K^\text{Airy}(\xi_1,\xi_2):=\frac{\mathrm{Ai}(\xi_1)\mathrm{Ai}'(\xi_2)-\mathrm{Ai}'(\xi_1)\mathrm{Ai}(\xi_2)}{\xi_1-\xi_2}.
	\end{equation}
	In addition, in the lower half plane, the $\mathcal J_n$ kernel (cf.~\eqref{jn}) for GUE
	\begin{equation}
		\mathcal J_n^\mathrm{GUE}(\lambda_1,\lambda_2)=-\frac{1}{h_{n-1}}\frac{\He_n(\lambda_1)\widehat{\He}_{n-1}(\lambda_2)-\He_{n-1}(\lambda_1)\widehat{\He}_n(\lambda_2)}{\lambda_1-\lambda_2}
	\end{equation}
	(cf.~\cite{Toda,FY}) gives the Airy--Bairy kernel as $n\to\infty$:
	\begin{equation} \label{dsl}
		\lim_{n\to\infty}n^{-1/6}e^{\frac{\lambda_1^2-\lambda_2^2}{4}}\mathcal J_n^\mathrm{GUE}(\lambda_1,\lambda_2)=2\pi i\omega \mathcal J^\mathrm{Ai-Bi}(\xi_1,\xi_2),
	\end{equation} where
	\begin{equation}
		\mathcal J^\mathrm{Ai-Bi}(\xi_1,\xi_2):=-\frac{\mathrm{Ai}(\xi_1)\omega\mathrm{Ai}'(\omega\xi_2)-\mathrm{Ai}'(\xi_1)\mathrm{Ai}(\omega\xi_2)}{\xi_1-\xi_2},
	\end{equation} and $\omega:=e^{2\pi i/3}$.
	We study application of~\eqref{dsl} and of the determinantal formulas~\eqref{maineq} and~\eqref{maineq1} in a subsequent publication.
\end{remark}

\section{Applications}

\subsection{KP integrability and affine coordinates} \label{sec-kp}

We first prove Lemma~\ref{exp-an}.

\begin{proof}[Proof of Lemma~\ref{exp-an}]
	The case $n=0$ is direct. For the case $n\ge1$, along the lines of the proof of~\cite[Lemma 4.4]{YZ}, we denote for $n\in\ZZ_{\ge0}$,
	\begin{align}
		&\phi_A(\lambda,n):=\lambda^{-n}\psi_A^\star(\lambda,n),\qquad\phi_B(\lambda,n):=h_n^{-1}\lambda^n\psi_B^\star(\lambda,n).
	\end{align}
	Then for $n\in\ZZ_{\ge1}$,
	\begin{equation}\label{an}
		\frac{\lambda_2^n}{\lambda_1^n}B_n(\lambda_1,\lambda_2)=\frac{\phi_A(\lambda_1,n)\phi_B(\lambda_2,n-1)-\gamma_n\frac{\phi_A(\lambda_1,n-1)}{\lambda_1}\frac{\phi_B(\lambda_2,n)}{\lambda_2}}{\lambda_1-\lambda_2},
	\end{equation}
	where we recall that $\gamma_n=\frac{h_n}{h_{n-1}}$ for $n\ge1$. Similar to~\cite{BDYds,DYZ,YZ}, we have
	\begin{align}
		\phi_B(\lambda_2)=&\phi_B(\lambda_1)+\phi'_B(\lambda_1)(\lambda_2-\lambda_1)+(\lambda_2-\lambda_1)^2\frac{\dif}{\dif\lambda_1}\left(\frac{\phi_B(\lambda_1)-\phi_B(\lambda_2)}{\lambda_1-\lambda_2}\right),\label{phi2}\\
		\frac{\phi_B(\lambda_2)}{\lambda_2}=&\frac{\phi_B(\lambda_1)}{\lambda_1}+\frac{\dif}{\dif\lambda_1}\left(\frac{\phi_B(\lambda_1)}{\lambda_1}\right)(\lambda_2-\lambda_1)+(\lambda_2-\lambda_1)^2\frac{\dif}{\dif\lambda_1}\left(\frac{\phi_B(\lambda_1)/\lambda_1-\phi_B(\lambda_2)/\lambda_2}{\lambda_1-\lambda_2}\right),\label{phi2'}
	\end{align} where we omit the parameter $n$. Since $\phi_A(\lambda,n)\sim1+O(\lambda^{-1})$ and $\phi_B(\lambda,n)\sim1+O(\lambda^{-1})$, substituting~\eqref{phi2},~\eqref{phi2'} into~\eqref{an}, and using~\eqref{wronskian}, we obtain the lemma.
\end{proof}

We now prove Corollary~\ref{KPthm}. 

\begin{proof}[Proof of Corollary~\ref{KPthm}]
	Using Theorem~\ref{main}, Lemma~\ref{exp-an} and the converse theorem of Zhou's theorem, we obtain the corollary.
\end{proof}

Denote
\begin{align}\label{pi}
	\pi_l(\lambda)=\sum_{k=0}^la_k(l)\lambda^{l-k},\qquad\lambda^l=\sum_{k=0}^lc_k(l)\pi_{l-k}(\lambda),
\end{align} and let
\begin{equation}
	T_{k,l}(n):=A_{k-1,l}(n)-A_{k,l-1}(n).
\end{equation}
We may prove a useful lemma.
\begin{lem} \label{at}
	We have
	\begin{align}
		&A_{k,0}(n)=c_{k+1}(n+k)~~(n\ge1),\qquad A_{0,l}(n)=-a_{l+1}(n)~~(n>l), \label{0-case}\\
		&T_{k,l}(n)=a_l(n)c_k(n-1+k)-a_{l-1}(n-1)c_{k-1}(n+k-1)\gamma_n~~(n>l). \label{klcase}
	\end{align}
\end{lem}

\begin{proof}
	By definition,
	\begin{equation}
		c_k(l+k)=\frac{1}{h_l}\int_{\RR}y^{l+k}\pi_l(y)\dif\mu(y).
	\end{equation}
	By comparing the coefficients in the two hand sides of~\eqref{affine}, we obtain the lemma.
\end{proof}

Now using Theorem~\ref{main}, Corollary~\ref{KPthm} and Lemma~\ref{at}, we prove Theorem~\ref{hankel}. For a matrix $M$, denote $M_{i_1,\cdots,i_m}^{j_1,\cdots,j_n}$ as the matrix obtained from the matrix $M$ deleting the $i_1$th, \dots, $i_m$th rows and the $j_1$th, \dots, $j_n$th columns. We denote
\begin{equation}
	\mathcal H:=H_{(0,1,\cdots,n;0,1,\cdots,n)}
\end{equation} We first state a lemma.

\begin{lem} \label{det}
	For any $2\le i<j\le n-1$,
	\begin{equation}
		\det(\mathcal H_{i,j}^{n,n+1})+\det(\mathcal H_{j,n+1}^{i-1,n+1})-\det(\mathcal H_{j-1,n+1}^{i,n+1})=0.
	\end{equation}
\end{lem}

\begin{proof}
	Denote by $\mathcal K$ the adjoint matrix of $\mathcal H$. Then from~\cite[Section 1.4]{Matrix}, we only need to prove that
	\begin{equation} \label{adj}
		\det\begin{pmatrix}
			\mathcal K_{n,i}&\mathcal K_{n,j}\\\mathcal K_{n+1,i}&\mathcal K_{n+1,j}
		\end{pmatrix}+\det\begin{pmatrix}
			\mathcal K_{i-1,j}&\mathcal K_{i-1,n+1}\\\mathcal K_{n+1,j}&\mathcal K_{n+1,n+1}
		\end{pmatrix}-\det\begin{pmatrix}
			\mathcal K_{i,j-1}&\mathcal K_{i,n+1}\\\mathcal K_{n+1,j-1}&\mathcal K_{n+1,n+1}
		\end{pmatrix}=0.
	\end{equation} 
	It is known~\cite{Bezoutian} (see also~\cite{Hankel}) that $\mathcal K$ is a Bezoutian matrix with respect to the orthogonal polynomials associated to moments $\{m_k\}_{k\ge0}$:
	\begin{equation}
		K_{n+1}(\lambda_1,\lambda_2)=\sum_{l=0}^{n}\frac{\pi_l(\lambda_1)\pi_l(\lambda_2)}{h_l}=\frac{1}{h_{n}}\frac{\pi_{n+1}(\lambda_1)\pi_n(\lambda_2)-\pi_n(\lambda_1)\pi_{n+1}(\lambda_2)}{\lambda_1-\lambda_2}=\frac{1}{\Delta_n}\sum_{i,j=0}^{n}\mathcal K_{i+1,j+1}\lambda_1^i\lambda_2^j.
	\end{equation} Thus
	\begin{equation}
		\mathcal K_{i,j}=\Delta_n\sum_{k=\max\{i-1,j-1\}}^n\frac{a_{k-i+1}(k)a_{k-j+1}(k)}{h_k}.
	\end{equation} 
	Then~\eqref{adj} is equivalent to
	\begin{align}
		&-\frac{a_{j-i-1}(j-2)}{h_{j-2}}-\frac{a_{n-i+1}(n)a_{n-j}(n-1)-a_{n-i}(n-1)a_{n-j+1}(n)}{h_{n-1}}\nn\\
		&\qquad+\sum_{k=j-1}^{n-1}\frac{a_{k-i+2}(k)a_{k-j+1}(k)-a_{k-i+1}(k)a_{k-j+2}(k)}{h_k}=0. \label{adj1}
	\end{align}
	
	We write
	\begin{equation}
		\alpha_k:=\frac{1}{h_k}(a_{k-i+2}(k+1)a_{k-j+1}(k)-a_{k-i+1}(k)a_{k-j+2}(k+1)).
	\end{equation} By~\eqref{three-term}, we have
	\begin{equation}
		a_{k-i}(k)=a_{k-i}(k+1)+\beta_ka_{k-i-1}(k)+\frac{h_k}{h_{k-1}}a_{k-i-2}(k-1).
	\end{equation}
	Thus
	\begin{align}
		&\frac{1}{h_k}(a_{k-i+2}(k)a_{k-j+1}(k)-a_{k-i+1}(k)a_{k-j+2}(k))\nn\\
		=&\frac{1}{h_k}\biggl(\left(a_{k-i+2}(k+1)+\beta_ka_{k-i+1}(k)+\frac{h_k}{h_{k-1}}a_{k-i}(k-1)\right)a_{k-j+1}(k)\nn\\
		&\qquad-a_{k-i+1}(k)\left(a_{k-j+2}(k+1)+\beta_ka_{k-j+1}(k)+\frac{h_k}{h_{k-1}}a_{k-j}(k-1)\right)\biggr)\nn\\
		=&\alpha_k-\alpha_{k-1}.
	\end{align}
	Then the LHS of~\eqref{adj1} writes
	\begin{equation}
		\alpha_{j-2}-\alpha_{n-1}+\sum_{k=j-1}^{n-1}(\alpha_k-\alpha_{k-1})=0.
	\end{equation}
	Thus the lemma is proved.
\end{proof}

Now we can prove Theorem~\ref{hankel}.
\begin{proof}[Proof of Theorem~\ref{hankel}]
	The orthogonal polynomials $\{\pi_k\}_{k\ge0}$ with respect to $\dif\mu(\lambda)$ read
	\begin{equation}
		\pi_k(\lambda)=\frac{1}{\Delta_{k-1}}\det\begin{pmatrix}
			m_0&m_1&\cdots&m_k\\
			m_1&m_2&\cdots&m_{k+1}\\
			\cdots&\cdots&\cdots&\cdots\\
			m_{k-1}&m_k&\cdots&m_{2k-1}\\
			1&\lambda&\cdots&\lambda^k
		\end{pmatrix}.
	\end{equation}
	Thus we have
	\begin{align}
		a_{n-k}(n)&=\frac{(-1)^{n-k}}{\Delta_{n-1}}\det H_{(0,1,\cdots,n-1;0,1,\cdots,\widehat{k},\cdots,n)},\\
		c_{n-l}(n)&=\frac{1}{\Delta_l}\det H_{(0,1,\cdots,l-1,n;0,1,\cdots,l)},\\
		\gamma_n&=\frac{\Delta_n\Delta_{n-2}}{\Delta_{n-1}^2}.
	\end{align}
	Therefore, $A_{0,j}(n)=-a_{j+1}(n)$ and $A_{i,0}(n)=c_{i+1}(n+i)$ satisfy~\eqref{hk}. Then by~\eqref{klcase},
	\begin{align}
		T_{i,j}(n)=&\frac{(-1)^j}{\Delta_{n-1}^2}\biggl(\det H_{(0,1,\cdots,\widehat{n-j},\cdots,n;0,1,\cdots,n-1)}\det H_{(0,1,\cdots,n-2,n+i-1;0,1,\cdots,n-1)}\nn\\
		&\qquad\quad+\det H_{(0,1,\cdots,\widehat{n-j},\cdots,n-1;0,1,\cdots,n-2)}\det H_{(0,1,\cdots,n-1,n+i-1;0,1,\cdots,n)}\biggr)\nn\\
		=&\frac{(-1)^j}{\Delta_{n-1}}\det H_{(0,1,\cdots,\widehat{n-j},\cdots,n;0,1,\cdots,n-2,n+i-1)},
	\end{align}
	where the last equation follows from the Sylvester's determinant identity~\cite[Section 2.3]{Matrix}:
	\begin{equation}
		\det(M)\det(M_{u,v}^{s,t})=\det(M_u^s)\det(M_v^t)-\det(M_u^t)\det(M_v^s),
	\end{equation}
	for any $u<v$ and $s<t$.
	
	Now we assume that $A_{i,j}(n),j<n-1$ satisfies~\eqref{hk}, we want to show that $A_{i-1,j+1}(n)$ also satisfies~\eqref{hk}. Then the proof is complete following the fact that $A_{i,0}(n)$ and $A_{0,j}(n)$ satisfy~\eqref{hk}. Since
	\begin{equation}
		A_{i-1,j+1}(n)=A_{i,j}(n)+T_{i,j+1}(n),
	\end{equation}
	we need to prove that
	\begin{align}
		0=&\det H_{(0,1,\cdots,\widehat{n-j-1},\cdots,n;0,1,\cdots,n-2,n+i-1)}-\det H_{(0,1,\cdots,\widehat{n-j-1},\cdots,n-1,n+i;0,1,\cdots,n-1)}\nn\\
		&-\det H_{(0,1,\cdots,\widehat{n-j-2},\cdots,n-1,n+i-1;0,1,\cdots,n-1)}\nn\\
		=&\sum_{\substack{2\le k\le n\\k\neq n-j}}(-1)^{k+n}m_{n+i+k-2}(\det(\mathcal H_{k,n-j}^{n,n+1})+\det(\mathcal H_{n-j,n+1}^{k-1,n+1})-\det(\mathcal H_{n-j-1,n+1}^{k,n+1})),
	\end{align}
	where we use the notation that $\det(\mathcal H_{j,i}^{n,n+1})=-\det(\mathcal H_{i,j}^{n,n+1})$ for $i<j$. This follows from Lemma~\ref{det}. Thus we obtain the theorem.
\end{proof}

\subsection{Properties of connected correlators} \label{section-rational}

We prove the following theorem.

\begin{thm} \label{rational}
		For any integers $i_1,\cdots,i_k\ge1$, $\langle\tr M^{i_1}\cdots\tr M^{i_k}\rangle_c(n)$ is a polynomial of $a_j(m)$, $1\le j\le i_1+\cdots+i_k$, $n-1\le m\le n+i_1+\cdots+i_k$. Here $a_j(m)$ are defined in~\eqref{pi}.
\end{thm}
\begin{proof}
	First, note that $A_{0,j}(n)=-a_{j+1}(n)$. Also by comparing the coefficients of $x^n$ and $x^{n-1}$ in the three-term recurrence relation~\eqref{three-term}, we have that
	\begin{align}
	\beta_n=&a_1(n)-a_1(n+1), \\
	\gamma_n=&a_2(n)-a_2(n+1)-\beta_na_1(n),
	\end{align}
	are polynomials in $a_j(m)$, $j=1,2$, $m=n,n+1$.
	
	Next, $A_{i,0}(n)=c_{i+1}(n+i)$, where $c_j(n)$ are defined in~\eqref{pi}. We prove that for any $j\ge0$, $c_{j}(n)$ is a polynomial in the coefficients of the orthogonal polynomials by induction on $j$. For $j=0$, $c_0(n)=1$. Assume that $c_k(n)$ is a polynomial in the coefficients of the orthogonal polynomials for any $k<j$. Then by definition,
	\begin{equation} \label{cnj}
		c_{j}(n)=-\sum_{k=0}^{j-1}c_{k}(n)a_{j-k}(n-k)
	\end{equation}
	is a polynomial in $a_j(m)$. Hence $A_{i,0}(n)$ is also a polynomial in $a_j(m)$, $1\le j\le i+1$, $n\le m\le n+i$.
	
	Finally, letting $A_{0,-1}(n)=-a_0(n)=-1$ and $A_{-1,0}(n)=c_0(n)=1$, for $k,l\ge1$,
	\begin{equation}
		T_{k,l}(n)=-A_{0,l-1}(n)A_{k-1,0}(n)+A_{0,l-2}(n-1)A_{k-2,0}(n+1)\gamma_n \label{tij}
	\end{equation}
	is a polynomial in $a_j(m)$, $1\le j\le\max\{2,k,l\}$, $n-1\le m\le\max\{n+1,n+k-1\}$. Thus by~\eqref{klcase}, $A_{k,l}(n)$ are polynomials in $a_j(m)$, $1\le j\le k+l+1$, $n-1\le m\le n+k+l$. Now by comparing the coefficients in~\eqref{maineq}, we obtain the theorem.
\end{proof}

As a direct corollary, we obtain:
	\begin{cor} \label{rational-cor}
		If the orthogonal polynomials $\{\pi_l\}_{l\ge0}$ satisfy that for any $k\ge0$, $a_k(n)$ is rational (resp. polynomial, holomorphic, meromorphic) in $n$, then for any $i_1,\cdots,i_k\ge1$, $\langle\tr M^{i_1}\cdots\tr M^{i_k}\rangle_c(n)$ is rational (resp. polynomial, holomorphic, meromorphic) in $n$.
	\end{cor}

\subsection{Examples}

\begin{example}[GUE: Gaussian unitary ensemble] \label{gue}
	GUE is defined on the space of hermitian matrices $\mathcal{H}_n$ with measure 
	\begin{equation}
		\dif\mu(\lambda)=e^{-\frac{1}{2}\lambda^2}\dif\lambda.
	\end{equation}
	The corresponding orthogonal polynomials are known~\cite{Szego} to be the Hermite polynomials
	\begin{equation} \label{orth-herm}
		\He_n(\lambda)=\sum_{m=0}^{[\frac{n}{2}]}\frac{(-1)^m}{m!2^m}(n-2m)_{2m+1}\lambda^{n-2m}.
	\end{equation}
	In this case, the affine coordinates are given in~\eqref{kpgue}, which are obviously polynomials in $n$. By~\eqref{orth-herm} and the proof of Theorem~\ref{rational}, we can also see that the affine coordinates are polynomials in $n$. In addition, by Corollary~\ref{rational-cor}, the connected correlators are also polynomials in $n$.
\end{example}

\begin{example}[LUE: Laguerre unitary ensemble~\cite{Lag,Sonine,Mehta,Szego,Laguerre,LUE}] \label{lue}
	LUE is defined on the space of $n\times n$ positive definite hermitian matrices $\mathcal{H}_n^+$ with measure
	\begin{equation}
		\dif\mu(\lambda)=\lambda^\alpha e^{-\lambda}\mathbf{1}_{(0,+\infty)}(\lambda)\dif \lambda,
	\end{equation}
	where $\Re(\alpha)>-1$. The corresponding orthogonal polynomials are known to be the generalized Laguerre polynomials
	\begin{equation} \label{orth-lue}
		L_n(\alpha;\lambda)=\sum_{k=0}^n\frac{(-1)^k}{k!}(n-k+1)_k(n+\alpha-k+1)_k\lambda^{n-k}.
	\end{equation}
	In this case, the explicit expressions of affine coordinates are derived in~\cite{ZhouGro,LUE}:
	\begin{equation}
		A_{i,j}^{\mathrm{LUE}}(n)=\frac{(-1)^j}{i!j!(i+j+1)}(n-j)_{i+j+1}(n+\alpha-j)_{i+j+1}.
	\end{equation}
	Obviously, they are polynomials in $n$. By~\eqref{orth-lue} and the proof of Theorem~\ref{rational}, we can also see that the affine coordinates are polynomials in $n$. In addition, by Corollary~\ref{rational-cor}, the connected correlators are also polynomials in $n$.
\end{example}

\begin{example}[JUE: Jacobi unitary ensemble~\cite{Jacobi,Mehta,Szego,JUE,corJUE}] \label{jue}
	JUE is defined on the space $\mathcal{H}_n(0,1)$ with measure
	\begin{equation}
		\dif\mu(\lambda)=\lambda^\alpha(1-\lambda)^\beta\mathbf{1}_{(0,1)}(\lambda)\dif\lambda,
	\end{equation}
	where $\Re(\alpha),\Re(\beta)>-1$. The corresponding orthogonal polynomials are known to be the Jacobi polynomials
	\begin{equation} \label{orth-jue}
		J_n(\alpha,\beta;\lambda)=\frac{n!}{(\alpha+\beta+n+1)_n}\sum_{k=0}^n\binom{n+\alpha}{k}\binom{n+\beta}{n-k}(\lambda-1)^k\lambda^{n-k}.
	\end{equation}
	In this case, the explicit expressions of affine coordinates are derived in~\cite{JUE,corJUE} (see~\cite{HE,Emergent,BY}):
	\begin{equation}
		A_{i,j}^\mathrm{JUE}(n)=\frac{(-1)^j}{(i+j+1)i!j!}\frac{(n-j)_{i+j+1}(n+\alpha-j)_{i+j+1}}{(2n+\alpha+\beta-j)_{i+j+1}}.
	\end{equation}
	Obviously, they are rational functions in $n$. By~\eqref{orth-jue} and the proof of Theorem~\ref{rational}, we can also see that the affine coordinates are rational functions in $n$. In addition, by Corollary~\ref{rational-cor}, the connected correlators are also rational functions in $n$.
\end{example}

\begin{example}[Atkin polynomials~\cite{KZ}] \label{atkin}
	Let $j=j(\tau)=q^{-1}+744+196884q+\cdots$ be the modular invariant, where $q=e^{2\pi i\tau}$. Let $\theta:[0,1728]\to[\pi/3,\pi/2]$ be the inverse of $\theta\mapsto j(e^{i\theta})$. Following~\cite{KZ}, consider the measure
	\begin{equation}
		\dif\mu(j)=\frac{6}{\pi}\theta'(j)\mathbf{1}_{[0,1728]}(j)\dif j.
	\end{equation}
	This gives the so-called Atkin polynomials $A_n(j)$ with three-term recurrence relation
	\begin{equation}
		jA_n(j)=A_{n+1}(j)+(\iota_{2n}+\iota_{2n+1})A_n(j)+\iota_{2n-1}\iota_{2n}A_{n-1}(j),\qquad n\ge0
	\end{equation} where
	\begin{equation}
		\iota_n=\begin{cases}
			0,&n=0,\\
			720,&n=1,\\
			12\left(6+\frac{(-1)^n}{n-1}\right)\left(6+\frac{(-1)^n}{n}\right),&n\ge2.
		\end{cases}
	\end{equation}
	For Atkin polynomials, $a_k(n)$ is rational in $n$ for each $k\ge0$ (see~\cite{KZ}). By the proof of Theorem~\ref{rational}, we can see that the affine coordinates are rational functions in $n$. In addition, by Corollary~\ref{rational-cor}, the connected correlators are also rational functions in $n$.
\end{example}

\section{Proof of Theorem~\ref{dual-ex}} \label{sec-dual}

In introduction, we mentioned the KP dual for hermitian matrix models which is characterized by
\begin{equation}
	\widetilde{\langle s_\rho\rangle}(n)=(-1)^{|\rho|}\langle s_{\rho'}\rangle(-n),\qquad n\ge\max\{2,|\rho|\}.
\end{equation}
Now we prove Theorem~\ref{dual-ex}.

\begin{proof}[Proof of Theorem~\ref{dual-ex}]
	The three-term recurrence coefficients associated to $\dif\mu$ read
	\begin{equation}
		\beta_n=\beta,n\ge0,\qquad\gamma_n=\begin{cases}
			\gamma,&n\ge1,\\
			0,&n=0.
		\end{cases}
	\end{equation}
	The associated orthogonal polynomials $\pi_n$ read
	\begin{equation}
		\pi_n(\lambda)=\sum_{k=0}^na_k(n)\lambda^{n-k},\qquad\lambda^n=\sum_{k=0}^nc_k(n)\pi_{n-k}(\lambda),
	\end{equation}
	where $a_0(n)\equiv1,c_0(n)\equiv1$, and for $k\ge1$,
	\begin{align}
		a_k(n)=&\sum_{j=0}^{[k/2]}(-1)^{k-j}\frac{\beta^{k-2j}\gamma^j}{j!(k-2j)!}(n-k+1)_{k-j}, \label{ak}\\
		c_k(n)=&(n-k+1)\sum_{j=0}^{[k/2]}\frac{\beta^{k-2j}\gamma^j}{j!(k-2j)!}(n-k+j+2)_{k-j-1}.
	\end{align}
	Let
	\begin{equation}
		\widetilde a_k(n):=c_k(-n-1+k).
	\end{equation}
	Then the polynomials
	\begin{equation}
		\widetilde\pi_n(\lambda):=\sum_{k=0}^n\widetilde a_k(n)\lambda^{n-k},\qquad n\ge0,
	\end{equation} satisfy the three-term recurrence relation
	\begin{equation}
		\lambda\widetilde\pi_n(\lambda)=\widetilde\pi_{n+1}(\lambda)+\widetilde\beta_n\widetilde\pi_n(\lambda)+\widetilde\gamma_n\widetilde\pi_{n-1}(\lambda),\qquad n\ge0,
	\end{equation}
	where for $n\ge0$,
	\begin{equation}
		\widetilde\beta_n=\beta,\qquad\widetilde\gamma_n=\begin{cases}
			\gamma+\delta_{n,1}\gamma,&n\ge1,\\
			0,&n=0.
		\end{cases}
	\end{equation}
	The corresponding measure is $\dif\widetilde\mu$.
	Let
	\begin{equation}
		\widetilde c_k(n):=a_k(-n-1+k).
	\end{equation}
	Then
	\begin{equation}
		\lambda^n=\sum_{k=0}^n\widetilde c_k(n)\widetilde\pi_{n-k}(\lambda).
	\end{equation}

	For the hermitian matrix models $M,\widetilde M$ associated to $\dif\mu,\dif\widetilde\mu$, the affine coordinates satisfy
	\begin{equation}
		\widetilde A_{0,i}(n)=-A_{i,0}(-n),\qquad\widetilde A_{i,0}(n)=-A_{0,i}(-n),\qquad |n|\ge i.
	\end{equation}
	Thus for $|n|\ge\max\{2,i,j\}$,
	\begin{equation}
		\widetilde T_{i,j}(n)=T_{j,i}(-n).
	\end{equation}
	Therefore, for $|n|\ge\max\{2,i,j\}$,
	\begin{equation}
		\widetilde A_{i,j}(n)=-A_{j,i}(-n).
	\end{equation}
	Then using Theorem~\ref{main}, for $|n|\ge\max\{2,i_1+\cdots+i_k\}$,
	\begin{align}
		\widetilde{\langle\tr M^{i_1}\cdots\tr M^{i_k}\rangle}_c(n)&=(-1)^k\langle\tr M^{i_1}\cdots\tr M^{i_k}\rangle_c(-n),\\
		\widetilde{\langle\tr M^{i_1}\cdots\tr M^{i_k}\rangle}(n)&=(-1)^k\langle\tr M^{i_1}\cdots\tr M^{i_k}\rangle(-n).
	\end{align}
	Now for partition $\rho=(\rho_1,\cdots,\rho_k)$, $l(\rho):=k,|\rho|:=\rho_1+\cdots+\rho_k$. Let $p_\rho(M):=\prod_{j=1}^{l(\rho)}\tr M^{\rho_j}$. Then
	\begin{equation}
		\widetilde{\langle p_\rho\rangle}(n) = (-1)^{l(\rho)}\langle p_\rho\rangle(-n),\qquad |n|\ge\max\{2,|\rho|\}.
	\end{equation}
	By Frobenius character formula~\cite[Section 1.7]{Schur}, the Schur polynomial of $\rho$ reads
	\begin{equation}
		s_\rho=\sum_{|\lambda|=|\rho|}\frac{\chi^\rho_\lambda}{z_\lambda} p_\lambda,
	\end{equation}
	where $\chi^\rho_\lambda$ is the character of the irreducible representation of the symmetric group $S_{|\rho|}$ on $\lambda$. By~\cite[Section 1.7]{Schur},
	\begin{equation}
		\chi_\lambda^\rho=\chi_\lambda^{\rho'}(-1)^{|\lambda|-l(\lambda)},
	\end{equation}
	where $\rho'$ is the conjugate partition of $\rho$. Then for $|n|\ge\max\{2,|\rho|\}$,
	\begin{align}
		\widetilde{\langle s_\rho\rangle}(n)=&\sum_{|\lambda|=|\rho|}\frac{\chi^\rho_\lambda}{z_\lambda} (-1)^{l(\lambda)}\langle p_\lambda\rangle(-n)\nn\\
		=&\sum_{|\lambda|=|\rho|}\frac{\chi^{\rho'}_\lambda (-1)^{|\lambda|-l(\lambda)}}{z_\lambda}(-1)^{l(\lambda)}\langle p_\lambda\rangle(-n)\nn\\
		=&(-1)^{|\rho|}\sum_{|\lambda|=|\rho'|}\frac{\chi^{\rho'}_\lambda}{z_\lambda}\langle p_\lambda\rangle(-n)\nn\\
		=&(-1)^{|\rho|}\langle s_{\rho'}\rangle(-n).
	\end{align}
	Thus $\widetilde M$ is the conjugate dual of $M$.
\end{proof}

\begin{remark}[Relations with Kerov map]
	The Jacobi fraction of $\dif\mu$ writes
	\begin{equation}
		J(t)=\dfrac{1}{1-\beta t-\dfrac{\gamma t^2}{1-\beta t-\dfrac{\gamma t^2}{\cdots}}}.
	\end{equation} Then
	\begin{equation}
		J(t)=\frac{1}{1-\beta t}E\left(\frac{t}{1-\beta t}\right),
	\end{equation} where
	\begin{equation}
		E(z):=\dfrac{1}{1-\dfrac{\gamma z^2}{1-\dfrac{\gamma z^2}{\cdots}}}.
	\end{equation}
	Consider the Kerov map~\cite{Kerov,Zhou2}
	\begin{equation}
		\mathcal{K}:J(t)\mapsto 1+t\frac{\dif}{\dif t}\log J(t).
	\end{equation} Then let $z:=\frac{t}{1-\beta t}$, we have
	\begin{align}
		\mathcal{K}(J)(t)=&1+\beta z+(1+\beta z)z\frac{\dif}{\dif z}\log E(z)\nn\\
		=&1+\beta z+(1+\beta z)(\mathcal{K}(E)(z)-1)\nn\\
		=&(1+\beta z)\mathcal{K}(E)(z)\nn\\
		=&\frac{1}{1-\beta t}\mathcal{K}(E)\left(\frac{t}{1-\beta t}\right).
	\end{align}
	Now $E(z)$ satisfies the relation
	\begin{equation}
		E(z)=\frac{1}{1-\gamma z^2E(z)},\qquad E(0)=1,
	\end{equation} which implies that
	\begin{equation}
		E(z)=\frac{1-\sqrt{1-4\gamma z^2}}{2\gamma z^2}.
	\end{equation}
	Then
	\begin{align}
		\mathcal{K}(E)(z)=&1+z\frac{\dif}{\dif z}\log E(z)\nn\\
		=&\frac{1}{\sqrt{1-4\gamma z^2}}\nn\\
		=&\frac{1}{1-2\gamma z^2E(z)}\nn\\
		=&\dfrac{1}{1-\dfrac{2\gamma z^2}{1-\dfrac{\gamma z^2}{1-\dfrac{\gamma z^2}{\cdots}}}}.
	\end{align}
	Thus 
	\begin{align}
		\mathcal{K}(J)(t)=\dfrac{1}{1-\beta t-\dfrac{2\gamma t^2}{1-\beta t-\dfrac{\gamma t^2}{1-\beta t-\dfrac{\gamma t^2}{\cdots}}}}.
	\end{align}
	This is exactly the Jacobi fraction of $\dif\widetilde\mu$.
\end{remark}

When $\beta=0,\gamma=1$, the corresponding Jacobi fractions are the generating functions of Dyck paths and of grand Dyck paths. When $\beta=1,\gamma=1$, the corresponding Jacobi fractions are the generating functions of Motzkin paths and of grand Motzkin paths. When $\beta=2,\gamma=1$, the corresponding Jacobi fractions are related to the generating functions of Catalan numbers of type A and of type B. When $\beta=3,\gamma=2$, the corresponding Jacobi fractions are related to the generating functions of big Schr\"oder paths and of central Delannoy paths. These examples are conjectured in~\cite{Zhou2,Zhou3}, which are now proved.

\end{document}